\documentclass[letterpaper]{article}

\usepackage{chngcntr}
\usepackage{framed}
\usepackage{url}
\usepackage{bbm}
\usepackage[final]{graphicx}
\usepackage{amsmath}
\usepackage{varwidth}
\usepackage{amssymb,amsthm}
\usepackage{algorithm}
\usepackage{algorithmic}
\usepackage{appendix}
\usepackage{hyperref}
\usepackage{bm}
\usepackage{cite}
\usepackage{enumerate}
\usepackage[papersize={8.5in,11in},top=1in,left=1in,right=1in,bottom=1in]{geometry}
\usepackage{xfrac}
\usepackage{sidecap}

\makeatletter
\def\@seccntformat#1{\@ifundefined{#1@cntformat}%
   {\csname the#1\endcsname\quad}  
   {\csname #1@cntformat\endcsname}
}
\let\oldappendix\appendix 
\renewcommand\appendix{%
    \oldappendix
    \newcommand{\section@cntformat}{\appendixname~\thesection\quad}
}
\makeatother

\newenvironment{lessspaceitemize*}%
  {\begin{itemize}%
  \vspace{-1mm}
    \setlength{\itemsep}{0pt}%
    \setlength{\parskip}{0pt}}%
    {\vspace{-1mm} \end{itemize}}

\newenvironment{lessspaceenum*}%
  {\begin{enumerate}%
  \vspace{-1mm}
    \setlength{\itemsep}{0pt}%
    \setlength{\parskip}{0pt}}%
  {\end{enumerate}}

\counterwithout{figure}{section}

\DeclareMathOperator*{\argmax}{arg\,max}

\newtheorem{theorem}{Theorem}[section]
\newtheorem{definition}[theorem]{Definition}
\newtheorem{lemma}[theorem]{Lemma}
\newtheorem{claim}[theorem]{Claim}
\newtheorem{proposition}[theorem]{Proposition}
\newtheorem{corollary}[theorem]{Corollary}
\newtheorem*{inf_thm}{(Centerpiece Result)}{\bfseries}{\itshape}

\newlength\myindent
\newlength\myindenttwo
\title{Price Doubling and Item Halving:\\ Robust Revenue Guarantees for Item Pricing}
\author{Elliot Anshelevich\footnote{Rensselaer Polytechnic Institute, \texttt{eanshel@cs.rpi.edu, sekars@rpi.edu.}} \and Shreyas Sekar$^*$}

\begin{document}

\newpage

\maketitle

\begin{abstract}
We study approximation algorithms for revenue maximization based on static \emph{item pricing}, where a seller chooses prices for various goods in the market, and then the buyers purchase utility-maximizing bundles at these given prices. We formulate two somewhat general techniques for designing good pricing algorithms for this setting: {\em Price Doubling} and {\em Item Halving}. Using these techniques, we unify many of the existing results in the item pricing literature under a common framework, as well as provide several new item pricing algorithms for approximating both revenue and social welfare. More specifically, for a variety of settings with item pricing, we show that it is possible to deterministically obtain a log-approximation for revenue and a constant-approximation for social welfare {\em simultaneously}: thus one need not sacrifice revenue if the goal is to still have decent welfare guarantees. 

The main technical contribution of this paper is a $O((\log m + \log k)^2)$-approximation algorithm for revenue maximization based on the Item Halving technique, for settings where buyers have XoS valuations, where $m$ is the number of goods and $k$ is the average supply. Surprisingly, ours is the first known item pricing algorithm with polylogarithmic approximations for such general classes of valuations, and partially resolves an important open question from the algorithmic pricing literature about the existence of item pricing algorithms with logarithmic factors for general valuations~\cite{balcanBM08}. We also use the Item Halving framework to form envy-free item pricing mechanisms for the popular setting of multi-unit markets, providing a log-approximation to revenue in this case.
\end{abstract}

\section{Introduction}
{\em Item pricing} lies at the core of most markets: a seller assigns prices for their goods or services, and buyers decide whether or not to pay the asked prices. While other economic paradigms are also known, most economic interactions still fall into this basic market model in which the seller determines the prices, and each buyer purchases their most desirable bundle of goods at the given prices. Moreover, although there are exceptions, a seller often cannot choose to charge different prices for different customers, i.e., the prices are {\em non-discriminatory} and do not change depending on the buyer. The algorithmic problem studied in this work, that of \emph{computing static item prices to maximize revenue}, occupies a central place in the study of markets (see for example \cite{briestK11,guruswamiHKKKM05}). Despite the substantial body of work on item pricing, surprisingly little is known about maximizing revenue for complex buyer valuations; we attempt to fill this conspicuous void in the literature by presenting new approximation algorithms for submodular and XoS functions, and in the process, partially address a long-standing open question from~\cite{balcanBM08}.



More generally, we model item pricing mechanisms in which the seller posts a price for each distinct good in the market, followed by a consumption stage in which buyers purchase utility maximizing bundles of goods at the fixed prices. Arguably, the appeal of item pricing stems from its simplicity: not only do such mechanisms allocate resources in a decentralized fashion, but they also do so without discriminating against the buyers or imposing a large cognitive overhead on them. We study two natural mechanisms that fall under this umbrella:
\begin{description}

\item[Simultaneous Mechanism] Once the prices are fixed, buyers arrive at the same time, and consume goods in parallel. The algorithm must ensure \emph{envy-freeness} without violating any of the supply constraints., i.e., each buyer receives a utility-maximizing bundle of goods.

\item[Sequential Mechanism] A relaxation of the simultaneous mechanism, where the buyers arrive in some order and purchase utility maximizing bundles subject to available supply.
\end{description}

\noindent \textbf{Landscape of Item Pricing Results} To obtain a better understanding of our results, it is important to place them in the context of the larger body of work in this field. 
In the simultaneous model, much previous work has been done in designing item pricing mechanisms for both maximizing social welfare~\cite{azevedo2013walrasian,gul1999walrasian,hsumrrv16} and revenue~\cite{balcanBM08,briestK06,grandoniR16,guruswamiHKKKM05}. Since welfare maximization for this model is equivalent to the problem of computing a Walrasian equilibrium pricing, minus the market clearing constraint, it is relatively well understood: for gross substitutes valuations, one can compute optimal envy-free prices, while there exist no item prices that approximate the optimum welfare better than a $\Omega(\sqrt{m})$-factor for submodular valuations. On the other hand, the revenue maximization problem was shown to be $\Theta(\log^{1-\epsilon} m)$-inapproximable even for simple classes of buyer valuations.

By resorting to sequential mechanisms, Feldman et al.~\cite{feldmanGL15} provided an item pricing mechanism that approximates the optimum welfare up to a constant factor, even for the somewhat general class of XoS valuations (see Section \ref{sec:model}). Not much is known about the revenue maximization problem with sequential buyers: for XoS valuations our understanding of this regime is largely incomplete owing to the super-logarithmic gap between the known lower bound of $\Omega(\log m)$ for submodular valuations, and the upper bound of $2^{\tilde{O}(\sqrt{\log m})}$\footnote{In the prior literature, the norm is to assume a unit-supply environment and express the approximation bounds only in terms of the number of goods $m$, or buyers $N$.} corresponding to the current best approximation algorithm for revenue maximization. Among other results, in this work we shrink this gap by presenting a $O(\log^2 m)$-approximation algorithm for revenue maximization~\cite{balcanBM08}. 


Finally, our work is motivated by concerns that are somewhat orthogonal to those prevalent in fields such as mechanism design. Specifically, we are less interested in constructing novel auctions for revenue maximization, and more so in gaining a fundamental understanding of how markets with item pricing can operate. Towards this end, as is common in the item pricing literature~\cite{guruswamiHKKKM05,cheungS08,feldmanGL16}, we eschew concerns such as incentive compatibility, and consider a full information model, where the seller is aware of each buyer's valuation function. The reader is asked to refer to the related work for comparisons to more general information models assumed in regard to item pricing.


\vskip 3pt\noindent\textbf{Revenue, Welfare, and Bicriteria Approximations~} A noteworthy aspect of this work is that, as in much of the previous literature, all of our approximation factors are derived by comparing the revenue obtained by our algorithm to the social welfare of the optimum allocation, $SW(OPT)$. While one could speculate that the choice of a better benchmark might result in improved approximations, measuring revenue in terms of social welfare  allows for our results to have some useful implications~\cite{anshelevichS16}. For instance, since the optimum welfare is an upper bound on the revenue obtained by any individually rational mechanism, our main result provides a $O(\log^2 m)$-approximation (for unit supply) to the performance of \emph{any} reasonable algorithm, for e.g., even a centralized non-item pricing mechanism that arbitrarily allocates goods and charges each buyer her exact valuation for the allocated bundle. Secondly, a distinguishing feature of this paper is our focus on (deterministic) bicriteria approximation algorithms that simultaneously maximize both revenue and social welfare. While many revenue maximization algorithms in the literature do lead to inadvertent bicriteria factors, the bounds obtained for welfare are trivial since any $\alpha$-approximation for revenue in terms of $SW(OPT)$ implies an $\alpha$-approximation for social welfare as well. On the contrary, the unifying framework that we develop allows us to derive explicit bicriteria approximations, where both the revenue and social welfare are comparable to their respective optimal lower bounds. For example, while the results of~\cite{guruswamiHKKKM05} immediately imply a $(O(\log m + \log k), O(\log m + \log k))$-approximation for (revenue,welfare) for unit-demand buyers, our framework provides an improved $(O(\log m + \log k), O(1))$-bicriteria approximation for the same setting (see Table~\ref{tab:results}). 

\subsection{Summary of Contributions}
In this paper, we consider a typical setting consisting of a set of $m$ goods in limited supply, with $k$ being the supply averaged over all of the goods in the market, and $k_{\max}$ being the maximum supply of any good. There are $N$ buyers in the market having combinatorial valuation functions, and under a given set of prices, each buyer purchases a subset of the available goods that maximizes her quasilinear utility, i.e., value derived minus price paid. 

\subsubsection*{Unifying Framework}

At a fundamental level, our first contribution is a general framework for revenue maximization based on two complementary algorithmic approaches, \emph{price doubling} and \emph{item halving}. The framework serves as a focal point for unifying a number of results in the pricing literature (see Table~\ref{tab:results}); we add additional value to these results by transforming them into bicriteria approximation algorithms that simultaneously guarantee good revenue and welfare. Secondly, we leverage the insights gained from our framework to design new approximations for revenue maximization in other settings.

\begin{table}[htb]
\centering
\begin{tabular}{|l|c|c|}
\hline
\textbf{Setting} & \textbf{Previous Results} & {\textbf{Our Results}}\\
& (Revenue) & $\{ \text{Revenue,Welfare} $\}\\
\hline
Unit-Demand & $O(\log m + \log k)$~\cite{guruswamiHKKKM05}  & $\{O(\log m + \log k), 3\} $\\
\hline
Gross Substitutes & -  & $\{O(\log m + \log k), 3\}  $ \\
\hline
Single-Minded & $O(\sqrt{m}\log k_{max})$~\cite{cheungS08}  & $\{O(\sqrt{m}\log k_{max}), O(\sqrt{m})\} $\\
\hline
Unlimited Supply & $\Theta(\log m + \log N)$~\cite{balcanBM08} &  $\{O(\log m + \log N), 2\} $\\
\hline
Bundle Pricing  & $O(\alpha^* \log N)$~\cite{feldmanGL16} & $\{O(\alpha^*\log N), 6\alpha^*\}$\\
\hline
Non Envy-Free Item Pricing  & $O(\log k_{max})$~\cite{elbassioniFS10} & $\{O(\log k_{max}), 6\}$\\
\hline
\end{tabular}
\caption{Bicriteria approximation factors obtained for various settings. For a given valuation class, $\alpha^*$ denotes the best known approximation factor for the allocation problem.}
\label{tab:results}
\end{table}

\noindent\textbf{Price Doubling: Revenue without Sacrificing Welfare in Simultaneous Mechanisms}
Surprisingly, we observe that a number of algorithms in the literature actually fall within our price doubling framework, allowing us to strengthen \emph{all of these results} by transforming them to bicriteria approximation algorithms, whose guarantees are listed in Table~\ref{tab:results}.  Moreover, we present a black-box transformation for item pricing from welfare maximization to revenue maximization for simultaneous mechanisms. Specifically, given black-box access to any $\alpha$-approximate simultaneous mechanism for welfare maximization, we present a poly-time reduction based on price doubling to compute item prices that ensure a $O(\alpha (\log m + \log k + \log \alpha))$-approximation to the optimal revenue, and simultaneously a $3\alpha$-approximation for welfare. The result is rather powerful since it holds for arbitrary combinatorial valuations as long as the black-box algorithm satisfies a `no local improvement' property.

\subsubsection*{Item Halving: New Approximation Algorithms for XoS and Multi-Unit Buyers}

The main technical contribution of this paper is a new approximation algorithm for revenue maximization based on our item halving framework; this novel framework allows us to prove quite a different set of results than the more usual price doubling techniques.

\begin{inf_thm}
For settings where buyers have XoS valuations, we can compute item prices in poly-time such that the sequential mechanism with these prices results in a \\$O((\log m + \log k)^2)$-approximation to the optimum revenue.
\end{inf_thm}

Note that XoS valuations are a strict generalization of submodular functions. The sequential mechanism that we design achieves the desired revenue guarantee for two models of buyer arrival: $(i)$ when the arrival is adversarial within a polynomially bounded set of buyer orders, and $(ii)$ when buyers arrive according to an arbitrary probability distribution. For the completely adversarial case, we present a non-poly-time algorithm with the same guarantee, that serves as an existence result.

\noindent{\em Key Features of our Result:} The above theorem provides a positive answer to the long standing open question of whether static item pricing could be used to obtain (poly-)logarithmic approximation factors for complex valuations, and in the process, significantly improves upon the previous upper bound of $2^{O(\sqrt{\log (mk) \log \log (mk)})}$~\cite{balcanBM08}. Prior to this, similar approximation factors were only known for more general pricing schemes such as bundle pricing~\cite{feldmanGL16} or discriminatory pricing~\cite{chakrabortyHK09}. In contrast, ours is the first known poly-logarithmic approximation algorithm based on non-discriminatory item pricing! Moreover, we come close to closing the gap between the upper and lower bounds for this setting, the latter being a $\Omega(\log m)$ factor for submodular valuations with unit supply $(k=1)$.

Finally, in Section \ref{sec:multi-item} we apply the same item halving machinery developed for XoS valuations to obtain a simultaneous mechanism with a $O(\log m)$-revenue guarantee for the popular class of {\em multi-unit valuations,} where a buyer's value depends only on the number of items that she receives. Previously, the best known result for this setting was a $O(\log m + \log N)$-approximation based on sequential mechanisms~\cite{balcanBM08}.

\subsection{Related Work}
Clearly, the current paper is closely associated with the growing body of work on item pricing and more general envy-free schemes~\cite{branzeiFMZ16,feldmanFLS12,FiatW09} for social welfare and revenue maximization. Item pricing for maximizing only welfare has traditionally been a sought after area of research owing to its ties to Walrasian equilibrium; since our focus is primarily on revenue, we refer the reader to~\cite{feldmanwelfarepricing,hsumrrv16, roughgardenT15} for more recent algorithmic perspectives on the subject. On the other hand, the problem of item pricing for revenue maximization has recently gained traction in computer science with respect to both sequential and simultaneous mechanisms. A steady stream of research has yielded near-optimal approximation algorithms for a variety of settings including but not limited to unit demand~\cite{briestK11,chenD14,guruswamiHKKKM05}, single minded~\cite{balcanB07,briestK06,cheungS08}, graph minded~\cite{grandoniR16}, multi-unit markets~\cite{branzeiFMZ16,feldmanFLS12}, and unlimited supply settings~\cite{balcanC10,briestK11}. One of the main contributions of this paper is a general framework that captures many of the above results, and provides a recipe for converting them into bicriteria approximations.

Despite the tremendous body of work on revenue maximization, ours is the first known poly-logarithmic approximation algorithm based on static, item pricing for complex valuations such as submodular and XoS functions. Partial exceptions include the $O(\log m)$-approximation for `simple submodular functions' in~\cite{balcanBM08} and the $O(\log^3 m)$-approximation in~\cite{chakrabortyHK09} for settings consisting of a large number of buyers having the exact same valuation function. Although our central result improves upon the upper bound of $2^{\tilde{O}(\sqrt{\log m})}$ from~\cite{balcanBM08} for XoS functions, it is pertinent to mention that their result holds for a more general information model where the seller is not aware of the buyer valuations. For such a general model, it is reasonable to expect that the only possible strategy would be to price items uniformly (common price), as is done in~\cite{balcanBM08,chakrabortyHK09}. The presence of a matching lower bound of $2^{\Omega(\sqrt{\log m})}$~\cite{chakrabortyHK09} for uniform pricing motivates the need for pricing different goods differently. For such a scheme, we argue that knowledge of buyer valuations is necessary since it allows us to quantify each good's relative value in the market.

As mentioned previously, this work differs considerably from the literature on \emph{sequential posted pricing} in multi-parameter mechanism design (e.g., \cite{adamczykBFKL15,caiDW16,chawlaHMS10}). With the notable exception of~\cite{feldmanGL15}, most of the pricing schemes in this area are non-anonymous, i.e., based on offering the same good to different buyers at different prices.  A more benign form of price discrimination that has recently exploded in popularity is \emph{dynamic pricing}~\cite{blumGMS11,chakrabortyHK09,huang2015welfare}, where the price offered in each round of a sequential mechanism is independent of the identity of the arriving buyer. On the contrary, our prices are static and do not change across rounds: this lack of price discrimination is an important force behind our results.
Finally, understanding the nature of revenue-welfare trade-offs has been an important research agenda in auctions~\cite{diakonikolasPPS12,likhodedovS04} but has not received much attention in posted pricing markets with general valuations such as the one in this paper.

\section{Model and Preliminaries} \label{sec:model}
We consider a market setting with a set $\mathcal{N}$ of $N$ buyers and set $\mathcal{I}$ of $m$ goods. Each buyer $j$ has a monotone valuation function $v_j:2^{\mathcal{\mathcal{I}}} \to \mathbb{R}^{+}$, whereas the seller is bound by a supply constraint on each good, i.e., at most $k_i$ units of good $i \in \mathcal{I}$ are available to be sold. We use $k := \frac{\sum_{i \in \mathcal{I}}k_i}{m}$ to denote the average supply of goods in the market. Given an allocation $\vec{S} =  (S_1, S_2, \ldots, S_N)$ of goods to the buyers, we will use $SW(\vec{S}) = \sum_{j \in \mathcal{N}}v_j(S_j)$ to represent the social welfare of this allocation. For the same allocation, let $N_i(\vec{S})$ denote the set of buyers who received a copy of good $i$ and $k_i(\vec{S}) = |N_i(\vec{S})|$. An allocation is \emph{feasible} only if it does not violate the supply constraints i.e., for each $i \in \mathcal{I}$, no more $k_i$ buyers receive a copy of this good. Finally, let $\theta(\vec{S}) = \sum_{i \in \mathcal{I}}k_i(\vec{S})$ denote the total number of items allocated to buyers in the given allocation $\vec{S}$.

In this work, we will primarily deal with the XoS class consisting of \emph{fractionally subadditive} valuation functions; this class includes, for example, all submodular functions. A valuation function $v$ is said to belong to this class if there exists a set of additive clauses $(a_1, \ldots, a_r)$ such that for any $T \subseteq \mathcal{\mathcal{I}}$, $v(T) = \max_{j=1}^{r} a_j(T)$. Each additive clause $a_j$ has a single value $a_j(i)$ for each $i \in \mathcal{N}$ so that for a set $T$ of agents, $a_j(T) = \sum_{i \in T}a_j(i)$. We also provide new approximation algorithms for the class of \emph{multi-unit valuations}, where a buyer's valuation depends only on the number of items that she receives, i.e., for all $S, T \subseteq \mathcal{I}$ such that $|S| = |T|$, $v(S) = v(T)$. Some of our secondary results are shown for other extremely standard classes of valuations that we formally define in Appendix~\ref{app:definitions}. Finally, as is common in almost all of the literature concerning combinatorial buyer valuations, we assume that we have black-box access to certain oracles that allow us to query the set functions. Specifically, we consider the following types of oracles with respect to a valuation $v$: $(i$) \emph{Value Oracle}, that when given a set $S \subseteq \mathcal{I}$ returns the value of $v(S)$, $(ii)$ \emph{Demand Oracle} that when queried with a vector of item prices $\vec{p}$ returns a set $S$ that maximizes the quantity $v(S) - \sum_{i \in S}p_i$, and $(iii)$ an \emph{XoS oracle} that for an XoS function $v$ and a set $T \subseteq \mathcal{I}$ returns the additive clause $a_l$ that maximizes $a_l(T)$.


\subsection{Pricing Mechanisms}
In its full generality, a pricing mechanism $\mathcal{M}$ is a system for allocating goods to the buyers and charging them payments. In the current work however, we restrict our attention to two types of mechanisms based on \emph{non-discriminatory item pricing}, where the seller fixes a single price per good in advance, and buyers purchase utility-maximizing bundles at the given prices. The price paid by each buyer equals the sum of item prices of the goods that she receives. 

\begin{enumerate}
\item \textbf{Simultaneous Mechanism.}
In this mechanism, the seller posts a price $p_i$ per good, buyers arrive simultaneously, and purchase a utility-maximizing subset $S_j\subseteq\mathcal{I}$ under the given prices. A pair $(\vec{p},\vec{S})$ is said to be a \emph{valid outcome} of the simultaneous mechanism if the following constraints are satisfied: $(i)$ for every buyer $j \in \mathcal{N}$, $S_j$ maximizes $v_j(S_j) - \sum_{i \in S_j}p_i$\} and thus the solution is \emph{envy-free}, and $(ii)$ at most $k_i$ buyers are allocated a copy of good $i$.

\item \textbf{Sequential Mechanism.} The sequential mechanism is very similar to the previous mechanism except that buyers arrive sequentially and once the supply of a good is exhausted, it is no longer available to future buyers. Specifically, this mechanism proceeds as follows:

\begin{lessspaceenum*}
\item Initially, all of the goods are available.
\item For each good $i\in \mathcal{I}$, the seller posts a single price $p_i$ and decides on the quantity $q_i \in [0,k_i]$ of this good to supply.
\item The buyers arrive in some order, and each buyer purchases her utility maximizing bundle from the set of available goods.
\item At any stage of the mechanism, if $q_i$ copies of good $i$ are sold (i.e, the good is sold out), then it is marked as unavailable.
\end{lessspaceenum*}

A triple $(\vec{p},\vec{q},\vec{S})$ is said to be a valid outcome of the sequential mechanism if by setting the prices $p_i$ and supply constraints $q_i$, the above mechanism can result in the allocation $\vec{S}$ for some arrival order of the buyers. In fact, we will usually omit $\vec{q}$ from this and simply denote an outcome by $(\vec{p},\vec{S})$, since if such an outcome is possible, then we can always set $q_i=k_i(\vec{S})$ to obtain the same outcome.
\end{enumerate}

We consider three arrival models for buyers in the sequential mechanism: $(i)$ \emph{partially adversarial}, where the buyer arrival order comes from a polynomially bounded set of arrival orders, $(ii)$ \emph{random arrival}, where buyers arrive according a distribution $D^{\pi}$, and finally $(iii)$ \emph{fully adversarial}, where any of the $N!$ arrival orders are possible. For the objective of revenue maximization, we are able to provide efficient computational mechanisms for the first two arrival models, and an exponential time algorithm for the fully adversarial case that serves as an existence result.

\noindent\textbf{Revenue and Surplus} Given a valid outcome $(\vec{p},\vec{S})$ of either mechanism, we define $Rev(\vec{p},\vec{S}) := \sum_{i \in \mathcal{I}}p_i k_i(\vec{S}) = \sum_{j \in \mathcal{N}}\sum_{i \in S_j}p_i$ to be the resulting revenue, and $Surp(\vec{p},\vec{S}) := \sum_{j \in \mathcal{N}}\{v_j(S_j) - \sum_{i \in S_j}p_i\}$ to be the aggregate surplus or buyer utility. It is not hard to see that $SW(\vec{S}) = Rev(\vec{p},\vec{S}) + Surp(\vec{p},\vec{S}).$

\subsection{Sufficient Conditions for Revenue and Bicriteria Approximations}
We now present the first part of our general framework for revenue maximization that helps us to identify solutions with good revenue and welfare properties. Specifically, we present results pertaining to a general sequence of pricing solutions that obey certain desirable properties that we term as \emph{charging properties}. The claims proved in this section are not technically difficult; the proofs are mostly straight-forward and can be found in the Appendix. The computationally challenging aspect is to actually design algorithms that compute a sequence of solutions obeying the properties spelled out below, as we do in Sections \ref{sec:pricedoubling}-\ref{sec:multi-item}. Once these solutions are computed, our lemmas simply help in showing that it is possible to leverage them to obtain good revenue.

\begin{definition}
\label{defn_framework1}
A sequence of pricing solutions $(\vec{p}^{(t)}, \vec{S}^{(t)})_{t=1}^{\gamma}$ is said to satisfy the \emph{charging} property if for all $t \in [1,\gamma-1]$,

$$SW(\vec{S}^{(t)}) - SW(\vec{S}^{(t+1)}) \leq \alpha Rev(\vec{p}^{(t)}, \vec{S}^{(t)}),$$

and $SW(\vec{S}^{(\gamma)}) \leq \alpha Rev(\vec{p}^{(\gamma)}, \vec{S}^{(\gamma)})$ for some $\alpha \geq 1$.

\end{definition}

The essence of Definition~\ref{defn_framework1} is that the `loss in welfare' as we traverse the solutions in a top-down manner can be \emph{charged} to the solution's revenue. Below we also define a similar notion that we refer to as the \emph{generalized charging property}, which is a strict generalization of Definition~\ref{defn_framework1}. It is, however, useful to make explicit the simpler charging property for two reasons: $(i)$ despite its apparent simplicity, this property is at the heart of a number of revenue maximization algorithms in the literature, and $(ii)$ we are able to show strong bicriteria guarantees for this model, which fail to hold upon generalization. 

\begin{claim}\label{clm_simplecharging}
Consider a sequence of pricing solutions $(\vec{p}^{(t)}, \vec{S}^{(t)})_{t=1}^{\gamma}$ that satisfy the charging property. Then, there exists $\ell \in [0, \gamma]$ such that

$$Rev(\vec{p}^{(\ell)}, \vec{S}^{(\ell)}) \geq \frac{SW(\vec{S^{(1)}})}{\gamma \alpha}.$$

Moreover, for any target welfare parameter $c \in [1, \gamma \alpha]$, there exists an $1 \leq \ell_2 \leq \gamma$ such that

$$SW(\vec{S}^{(\ell_2)}) \geq \frac{SW(\vec{S^{(1)}})}{c}~~~\text{and}~~~Rev(\vec{p}^{(\ell_2)}, \vec{S}^{(\ell_2)}) \geq (1-\frac{1}{c})\frac{SW(\vec{S^{(1)}})}{\gamma \alpha}.$$
\end{claim}

For the purposes of intuition, $SW(\vec{S^{(1)}})$ is typically close to the optimum welfare and $\alpha$ is a small constant, whereas $\gamma$ is logarithmic in one of the market parameters such as $N$ or $m$.

\paragraph{Generalized Revenue-Welfare Charging Property}
Although the simple charging property presented above captures a number of revenue maximization algorithms in the item pricing literature, it is not enough to provide good algorithms for XoS valuations in the sequential setting, as we do in this paper. To do this, we provide a strict generalization of this property in a manner that makes it ideal to design pricing mechanisms which have black-box access to an allocation algorithm.


\begin{definition}
\label{defn_framework2}
Consider a sequence of benchmark allocations $\vec{B}^{(1)}, \vec{B}^{(2)}, \ldots, \vec{B}^{(\gamma)}$, and pricing solutions $(\vec{p}^{(t)}, \vec{S}^{(t)})_{t=1}^{\gamma}$. The sequence of solutions is said to satisfy the \emph{generalized charging} property if $(i)$ for all $t \in [1,\gamma-1]$,

$$SW(\vec{B}^{(t)}) - SW(\vec{B}^{(t+1)}) \leq \alpha Rev(\vec{p}^{(t)}, \vec{S}^{(t)}) + \frac{SW(\vec{B}^{(t)})}{\beta},$$

for $\alpha,\beta \geq 1$;

\noindent $(ii)$ $SW(\vec{B}^{(\gamma)}) \leq \alpha Rev(\vec{p}^{(\gamma)}, \vec{S}^{(\gamma)})$.

\end{definition}

When $\vec{S}^{(t)}=\vec{B}^{(t)}$ and $\beta$ is large, this property implies the simpler charging property above. Note that in the above definition, the benchmark allocations $(\vec{B}^{(t)})_{t=1}^\gamma$ are completely arbitrary, and may not actually be obtainable by any pricing mechanism.

\begin{claim}
\label{clm_frame2rev}
Consider a sequence of benchmark allocations $\vec{B}^{(1)}, \ldots, \vec{B}^{(\gamma)}$, and pricing solutions $(\vec{p}^{(t)}, \vec{S}^{(t)})_{t=1}^{\gamma}$ that satisfy the generalized charging property as per Definition~\ref{defn_framework2}. Then, there exists $\ell \in [1,\gamma]$ such that

$$Rev(\vec{p}^{(\ell)}, \vec{S}^{(\ell)}) \geq \frac{1}{\gamma\alpha}SW(\vec{B}^{(1)})\{1-\frac{\gamma-1}{\beta}\}.$$

\end{claim}

Unlike for the simple charging property, here we provide a weaker bicriteria trade-off result between revenue and welfare that can be summed up as follows: `either the revenue guarantee is much better than that stated in Claim~\ref{clm_frame2rev} or one can compute a single pricing solution with similar revenue guarantees as before but whose social welfare is only a constant factor away from $SW(\vec{B}^{(1)})$'. For a specific application of this somewhat abstract claim, see Claim~\ref{thm_bicriteriaxos}.

\begin{claim}
\label{clm_frame2bicrit}
\textbf{Bicriteria Bounds}
Suppose that we have a sequence of benchmark allocations and prices as defined in Definition~\ref{defn_framework2}. In addition, suppose that the following new property is satisfied for all $t$: $SW(\vec{B}^{(t)})(1-\frac{1}{\beta}) \leq \alpha Rev(\vec{p}^{(t)}, \vec{S}^{(t)}) + Surp(\vec{p}^{(t)}, \vec{S}^{(t)}).$ Then, there exists $\ell \in [0,\gamma]$ and some constant $c \geq 1$ that depends on the instance such that,

\begin{align*}
(i) SW(\vec{S}^{(\ell)}) & \geq SW(\vec{B}^{(1)})(1-\frac{\gamma-1}{\beta})\{\frac{1}{2}(1-\frac{1}{\beta}) - \frac{c}{2\gamma} + \frac{c}{2\gamma\alpha}\}\\
(ii) Rev(\vec{p}^{(\ell)}, \vec{S}^{(\ell)}) & \geq \frac{c}{2\gamma \alpha}SW(\vec{B}^{(1)})(1-\frac{\gamma-1}{\beta}).
\end{align*}
\end{claim}

\section{Price Doubling and Bicriteria Approximations}
\label{sec:pricedoubling}
The main technical contribution of this work is a new $O((\log m + \log k)^2)$-approximation algorithm for revenue maximization with XoS valuations, which is presented in Section~\ref{sec:itemhalving}. Before heading to the technical part, however, we present a few fundamental results for revenue and welfare maximization that illustrate the power of our unifying framework. The framework that we derive in this section, henceforth referred to as the price doubling framework, is based on generating a sequence of pricing solutions such that the prices across successive solutions differ by a multiplicative factor (usually two), and then applying the charging property from Definition~\ref{defn_framework2} to these solutions. Somewhat surprisingly, we observe that our framework is found lurking beneath the hood of a number of revenue maximization algorithms in the literature. These include item pricing algorithms for environments such as unit-demand~\cite{guruswamiHKKKM05} and single-minded valuations~\cite{cheungS08}, unlimited supply markets~\cite{balcanBM08} but also mechanisms based on more general pricing schemes such as \emph{envy-free bundle pricing}~\cite{feldmanGL15} and \emph{non envy-free item pricing}~\cite{elbassioniFS10}. Our simple observation allows us to immediately extend \emph{all} of these algorithms to bicriteria (revenue,welfare) approximations without any loss in the asymptotic factors for revenue.

%
%

\subsection{Reduction from Welfare to Revenue for Simultaneous Mechanisms}
Our first, and most general result in this section involves a reduction from revenue to welfare maximization based on price doubling, which yields not just the bicriteria algorithms mentioned previously, but also new revenue approximations for other settings (see Theorem~\ref{thm:implications} and Table \ref{tab:results}). Informally, the result indicates that in order to design item pricing mechanisms with good revenue, it is sufficient to focus our attentions on mechanisms that maximize social welfare. Our algorithm can be viewed as a generalization of a similar black-box transformation from revenue to Walrasian equilibrium for unit-demand buyers in~\cite{guruswamiHKKKM05}. However, the extension is quite non-trivial owing to the generality of our result, as it holds for settings with `arbitrary combinatorial valuations' under some small assumptions.

We now define some notions that are pertinent to our result. We say that a buyer has a (monotone) \emph{single-item valuation} if and only if $v(\{i\}) > 0$ for at most one good $i \in \mathcal{I}$. Secondly, an allocation, and by extension an algorithm, is said to be \emph{locally welfare maximizing} if for any buyer $j$ who is not allocated any goods, and any set of goods $T$ that are not \emph{sold out}, $v(T) = 0$. For instance, observe that any mechanism that always outputs a socially optimal solution trivially satisfies this property. The following theorem takes as input a black-box algorithm for a simultaneous item pricing mechanism whose social welfare is an $\alpha$-approximation to $SW(OPT)$ and creates an item pricing algorithm for revenue with a $\log$ factor increase in the approximation factor. Moreover, the resulting solution will still have good welfare at the same time.

\begin{theorem}\label{thm.black-box}
Let $\mathcal{V}$ be a class of buyer valuations that includes single-item valuations. Suppose that we are given black-box access to a locally welfare maximizing algorithm $Alg$, that provides an $\alpha$-approximate pricing solution for the problem of designing a simultaneous mechanism for welfare maximization for instances where buyer valuations belong to $\mathcal{V}$. Then, we can efficiently compute prices such that the simultaneous mechanism with these prices provides a $O(\alpha (\log(\alpha) + \log(m) + \log(k)))$-approximation to revenue and a $(3\alpha)$-approximation to social welfare.
\end{theorem}

We remark that the logarithmic gap between welfare and revenue is known to be tight even when all buyers have single item valuations~\cite{guruswamiHKKKM05,feldmanGL15}.

\noindent\textbf{Black-box Reduction} Before constructing our black-box reduction, we require some pertinent notation. Given an instance $G$ of our problem and a fixed `reserve price' $r \geq 0$, define a new instance $G(r)$ consisting of the same set of goods as $G$ and an augmented set of buyers comprising of the original buyer set $\mathcal{N}$ and a new set $\mathcal{N}'$ of $(m+1)k$ dummy buyers. Specifically for each good $i \in \mathcal{I}$, we define $k_i+1$ dummy buyers all of whom have single item valuations with $v(\{i\}) = r$ (and so $v(\{i'\}) = 0 ~ \forall i' \neq i$). Assuming that the original instance $G$ is clear from the context, we define $(\vec{p}(r), \vec{S}(r))$ to be the outcome of the algorithm $Alg$ for the instance $G(r)$ without any of the dummy buyers, i.e., we run $Alg$ for input $G(r)$ and set $\vec{p}(r)$ to be the prices output by this algorithm and $\vec{S}(r)$ to be sub-allocation output by this algorithm consisting only of the original set of buyers $\mathcal{N}$. Now, we define our black-box reduction as follows for $\gamma = 1+\log(\alpha m k)$, and $SW^0 = SW(Alg(G))$, which is the social welfare of the allocation returned by $Alg$ for the original instance.

\begin{lessspaceenum*}
\item For $t=1$ to $t = \gamma$
\item Set reserve price $r_t = 2^{t-1} \frac{SW^0}{2 mk}$.
\item Define $(\vec{p}^{(t)}, \vec{S}^{(t)}) = (\vec{p}(r_t), \vec{S}(r_t))$, i.e. the outcome of $Alg$ with the reserve price $r_t$.
\item Return the smallest index $\ell \in [1,\gamma]$ at which $Rev(\vec{p}^{(t)}, \vec{S}^{(t)}) \geq \frac{SW(\vec{S}^{(1)})}{6\gamma}$.
\end{lessspaceenum*}

The proof proceeds by showing that the sequence of solutions $(\vec{p}^{(t)}, \vec{S}^{(t)})_{t=1}^{\gamma}$ satisfies the simple charging property, and therefore, the bicriteria result follows from Claim~\ref{clm_simplecharging} with $c=\frac{3}{2}$. Details can be found in the Appendix. Note that we could have presented this theorem in the form of a continuous trade-off (by varying the value of $c$) quantifying the amount of revenue that the seller has to sacrifice in order to obtain the target welfare guarantee.

\subsection{Implications of the Black-Box Reduction and Charging Property}
\label{sec:implications}
The next theorem of this paper is a collection of bicriteria approximation algorithms that simultaneously maximize both revenue and welfare for a number of settings, not just restricted to item pricing. Some of these results are a consequence of Theorem~\ref{thm.black-box}, whereas others follow from previous papers~\cite{balcanBM08,cheungS08,elbassioniFS10,feldmanGL16}, where the algorithms implicitly used the charging property to obtain the revenue bounds. Before presenting the theorem, we define two mechanisms that strictly generalize (relax) the simultaneous item pricing mechanism.
\begin{description}
\item[Simultaneous bundle pricing]~\cite{feldmanGL16,rubinstein16} The seller partitions the set of goods into bundles, posts one price per bundle, and each buyer purchases the bundle of bundles that maximize her utility under the given prices.

\item[Non Envy-Free Item Pricing]~\cite{elbassioniFS10} The seller still posts a single price per good, but has the power to allocate non utility maximizing bundles to users as long as the total amount paid by any buyer is no larger than her valuation for the bundle received.
\end{description}




\begin{theorem}\label{thm:implications}
We can compute the following $\{\text{revenue, welfare}\}$ bicriteria approximations in poly-time for limited supply settings
\begin{enumerate}
\item A $\{O(\log m + \log k), 3\}$-approximation when buyers have gross substitutes valuation.

\item A $\{O(\sqrt{m}\log(k_{max})), O(\sqrt{m})\}$-approximation for single-minded valuations based on the algorithm of~\cite{cheungS08}.

\item A $\{O(\alpha^*\log(N)), 6\alpha^*\}$-approximation for arbitrary combinatorial valuations based on simultaneous bundle pricing~\cite{feldmanGL16} and a $\{O(\log(k_{max})), 6\}$-approximation algorithm based on non envy-free item pricing~\cite{elbassioniFS10} when buyers have subadditive valuations.
\end{enumerate}
Here $k_{max}$ refers to the maximum available supply of any good, and $\alpha^*$  denotes the approximation factor corresponding to the best known algorithm for the welfare maximizing allocation problem for the given class of valuations\footnote{For example $\alpha^* = \frac{e}{e-1}$ for submodular, XoS valuations~\cite{dobzinskiNS10} and~$\alpha^*=\sqrt{m}$ for general valuations~\cite{kolliopoulosS04}} .

\end{theorem}

The definitions of all of these classes of valuations are provided in Appendix~\ref{app:definitions}. The result for unit-demand buyers in Table~\ref{tab:results} follows from the fact that unit demand valuations satisfy the gross substitutes property. In addition, for settings with unlimited supply, it is not hard to see that one can obtain a $O(\log(N) + \log(m)), O(1))$-bicriteria approximation either by observing that the results in~\cite{balcanBM08} satisfy the charging property or by leveraging Theorem~\ref{thm.black-box} with a forced supply constraint of $N$ on each good. To the best of our knowledge, our results are the first such approximation algorithms for settings having gross substitutes valuations. Finally, we reiterate a point that we made in the Introduction. Any $c$-approximation algorithm for revenue obtained in terms of the optimum social welfare is trivially a $(c, c)$-bicriteria approximation. What makes Theorem~\ref{thm:implications} appealing is that the bound for welfare is much better than $c$, and in most cases, only a constant factor away from the optimum social welfare.

\section{Revenue Maximization via Item Halving for XoS Buyers}
\label{sec:itemhalving}
In this section, we present the central result of this paper: a computationally efficient, sequential posted pricing mechanism that achieves a $O((\log m + \log k)^2)$-approximation to the optimum revenue. In fact, as discussed previously, our approximation factor is obtained with respect to the social welfare maximizing allocation. Depending on the arrival order, our $O((\log m + \log k)^2)$ sequential mechanism comes in two flavors: the partially adversarial model, where the arrival order is adversarial within a bounded set of arrival orders, and the random model, where the buyers arrive according to an arbitrary distribution $D^{\Pi}$.  
%

\subsubsection*{Revenue Maximization for XoS Buyers: Challenges and Techniques}
How does one go about pricing items in unit supply for buyers having complex valuations? The state of the art in item pricing is based on the price doubling technique discussed in the previous section. Specifically, powerful approximation algorithms have resulted from the simple observation made in~\cite{balcanBM08} that offering a `single buyer' a random price (on all goods) from the set $\{2^t \frac{H}{m^2}\}_{t=1}^{1+\log^2 m}$, for a suitable choice of $H$, is sufficient to ensure good revenue in expectation. For instance, Chakraborty et al.~\cite{chakrabortyHK09} successfully leveraged this idea to design a dynamic pricing algorithm that obtains a polylog-approximation for revenue by resetting the price uniformly at random for every single buyer. Unfortunately, it is known that static pricing based on this technique could be rather sub-optimal ($2^{\Omega(\sqrt{\log m})}$) even for two buyers, since the earlier buyer could have large valuations for low revenue goods intended for the later buyer. The natural next step is to then consider pricing items differently based on the value they generate for buyers in some allocation, for e.g., as was done in~\cite{feldmanGL15}. However, this need not result in good revenue either for similar reasons. Therefore, we are forced to infer that current pricing techniques including price doubling are ineffectual in a static setting since they do not take into account the dependencies introduced by limited supply~\cite{chakrabortyHK09}.

The crucial observation that powers our centerpiece result is that in order to obtain good revenue, not only should the pricing be sensitive to the valuations, but it should also be sensitive to specific types of allocations whose sub-allocations are strongly undesirable to the buyers. Pricing items based on the value that it gives buyers in such `special allocations' forces the buyers to steer clear of sub-optimal purchases, thereby ensuring good revenue. Concretely, we reduce the problem of designing sequential mechanisms with good revenue to that of computing an allocation $\vec{A}$ such that for every sub-allocation $\vec{S}$ that only allocates half the number of items as $\vec{A}$ (i.e., $\theta(\vec{A}) \geq 2\theta(\vec{S})$), its social welfare is at least a $\alpha$ factor smaller than that of $\vec{A}$. 

How do we compute such a `special allocation'? For this, we turn to the item halving framework that takes some allocation as input and computes a sequence of benchmark allocations $(\vec{B}^{(t)})_{t=1}^{\gamma}$ such that in successive allocations, the number of items assigned to the buyers is halved and yet, the two allocations are close in terms of social welfare. Eventually, the framework must hit an impasse, either we compute the desired special allocation or we simply end up with one good, for which revenue maximization is trivial. At a more fundamental level, we couple our benchmark solutions with item prices and show that the sequence of solutions must satisfy the charging property from Definition~\ref{defn_framework2}. This allows us to efficiently identify a pricing solution $(\vec{p}^{(\ell)},\vec{S}^{(\ell)})$ such that the sequential mechanism with these prices yields a good approximation to the input allocation. Finally, looking beyond the unit-supply setting, our algorithm must also handle the fact that different copies of a good must be given the same price even though they may generate vastly different values for different buyers.


\subsection{Item Halving Framework}
\textbf{Notation} We will implicitly assume that we are dealing with XoS valuations for the rest of the section. Given an allocation $\vec{S} = (S_1, S_2, \ldots, S_N)$, we use $x_j^{S_j}$ to denote the XoS clause that maximizes buyer $j$'s valuation for the set $S_j$ of goods, i.e., $x_j^{S_j}(S_j) = v_j(S_j)$. Then, we can define the total value derived by buyers from any given good $i \in \mathcal{I}$ with respect to an allocation $\vec{S}$ as $U_i(\vec{S}) = \sum_{j \in N_i(\vec{S})}x_j^{S_j}(i)$. Additionally, for any integer $1 \leq r \leq k_i(\vec{S})$, we define $Top_i(r,\vec{S}) \subseteq N_i(\vec{S})$ to be the set of $r$ buyers in $N_i(\vec{S})$ with the highest values of $x^{S_j}_j(i)$.

Now, we are ready to define our first black-box algorithm that takes as input an allocation $\vec{A}$ and a parameter $\gamma \geq 1$, and computes a series of benchmark allocations $(\vec{B}^{(1)}, \ldots, \vec{B}^{(\gamma)})$ as well as pricing solutions $((\vec{p}^{(1)},\vec{S}^{(1)}) \ldots, (\vec{p}^{(\gamma)},\vec{S}^{(\gamma)}))$ that are in accordance with Definition~\ref{defn_framework2}. We present our black-box reduction (the core algorithm) in two parts, where the first part handles the main reduction and the computation of successive benchmark allocations, and the second part handles the boundary or tail condition, where a buyer has a large valuation for a single item. Note that the purpose of our core algorithm is to simply compute prices (and supply limits) for our eventual poly-logarithmic sequential mechanism.

\begin{figure}[ht]
\begin{framed}
\begin{enumerate}
\item \textbf{Input:} Allocation $\vec{A}$, parameter $\gamma \geq 1$. 
\item Define $\vec{B}^{(1)} = \vec{A}$.
\item For $t=1$ to $t=\gamma-1$
\item Set $\vec{p}^{(t)} = Prices(\vec{B}^{(t)},\gamma)$.~~~[See Figure \ref{fig.prices}]
\item Let $\pi$ be an arrival order chosen by an arbitrary function $\Gamma(\vec{p}^{(t)}, \vec{B}^{(t)})$ (to be fixed later)
\item Let $\vec{S}^{(t)}$ be an allocation obtained upon running the sequential mechanism with prices $\vec{p}^{(t)}$, supply $\vec{k}(\vec{B}^{(t)})$, and buyer arrival order $\pi$.
\item (Case I) if $2\theta(\vec{S}^{(t)}) \leq \theta(\vec{B}^{(t)})$, i.e., less than half the allocated goods from $\vec{B}^{(t)}$ are sold in $\vec{S}^{(t)}$: ~~~ set $\vec{B}^{(t+1)} = \vec{S}^{(t)}$.
\item (Case II) else:\\
~~~~~~~~~~~~set $\vec{B}^{(t+1)}=\text{Alloc-Unsold}(\vec{B}^{(t)},\vec{S}^{(t)})$.~~~[See Figure \ref{fig.unsold}]
\end{enumerate}
\end{framed}
\caption{Core Algorithm: Part I}
\label{fig:coremain}
\end{figure}

\begin{figure}[h]
\begin{framed}
\begin{enumerate}
\item Define $j^*$, $i^*$ to be a buyer, good pair that maximizes $v_{j}(\{i\})$.
\item Set $\vec{p}^{(\gamma)}_{i^*} = v_{j^*}(\{i^*\}) - \epsilon$, for some sufficiently small $\epsilon > 0$, and $\vec{p}^{(\gamma)}_{i} = \infty$, for $i \neq i^*$.
\item Let $\vec{S}^{(\gamma)}$ be the allocation in which $j^*$ receives the item $i^*$, and no one else receives anything. 
\end{enumerate}
\end{framed}
\caption{Core Algorithm Tail Conditions} \label{fig:tail}
\end{figure}

\paragraph{Description of Core Algorithm}
The core algorithm presented here (Figures~\ref{fig:coremain},~\ref{fig:tail}) is the main workhorse behind all of our results. The algorithm makes use of three separate functions. $Prices(\vec{S},\gamma)$ computes the next set of prices for each stage of our mechanism. $\Gamma$ selects an arbitrary arrival order, and thus determines the next allocation $\vec{S}^{(t)}$: everything we prove in this section will hold for arbitrary $\Gamma$, while in Section \ref{sec:applying} we will show how to formulate $\Gamma$ in order to create solutions with good revenue for adversarial and random arrival orders. For example, in the adversarial case $\Gamma$ should choose the arrival order that will result in the worst possible revenue; in this case when we show that some solution $S^{(\ell)}$ has high revenue this will also be true for all other allocations that could be created by other arrival orders. Finally, $\text{Alloc-Unsold}(\vec{B},\vec{S})$ allocates the unsold goods from the mechanism to the original buyers. The modular form of our black-box reduction provides us with the flexibility to define different variants based on the specific setting at hand. Beginning with the input solution as the first benchmark allocation, the algorithm computes prices $\vec{p}^{(t)}$ for each $t$ and then executes the sequential posted pricing mechanism with these prices and supply limits that equal the number of copies of each good allocated in the corresponding benchmark solution. A crucial decision is made in each iteration based on the number of items sold by the sequential mechanism. If this number is small, i.e., less than half of the items allocated in the benchmark allocation, the solution returned by the mechanism is carried over to the following round as the next benchmark solution, i.e., $\vec{B}^{(t+1)} = \vec{S}^{(t)}$. Otherwise, we instead allocate all the items which were sold in $\vec{B}^{(t)}$ but {\em unsold} in $\vec{S}^{(t)}$ using the Alloc-Unsold function to form the next benchmark solution. Thus, at each iteration the number of items allocated in the benchmark solution reduces by at least a factor of two. Finally, in the second half of the core algorithm (tail condition), we return a solution maximizing revenue subject to the fact that exactly one item is sold. We remark that our algorithm requires access to a demand oracle and XoS oracle for each buyer's valuation function, which are standard assumptions when dealing with XoS valuations.


\begin{figure}[h]
\begin{framed}
Return $\vec{p}$ such that:
\begin{align*}
 p_i & = \frac{U_i(\vec{S})}{2\gamma k_i(\vec{S})} & \text{if} ~~~ k_i(\vec{S}) \neq 0\\
 p_i & = \infty & \text{otherwise} & .
\end{align*}
\end{framed}
\caption{Function $Prices(\vec{S}, \gamma$)}
\label{fig.prices}
\end{figure}


\begin{figure}[htb]
\begin{framed}
Construct an allocation $\vec{C}$ as follows:
\begin{enumerate}
\item For each good $i \in \mathcal{I}$:
\item Define $q_i := k_i(\vec{B}) - k_i(\vec{S})$.
\item if $q_i > 0$:
\item Allocate item $i$ to each of the buyers in $Top_i(q_i,\vec{B})$.
\end{enumerate}
Return the allocation $\vec{C}$.
\end{framed}
\caption{Function $\text{Alloc-Unsold}(\vec{B},\vec{S})$}
\label{fig.unsold}
\end{figure}

We now prove the main theorem of this section, which allows us to use the charging property for the allocations generated by the above algorithm. In Section \ref{sec:applying}, we show how to apply this algorithm and the charging property to actually produce mechanisms with high revenue.

\begin{theorem} For $\gamma=\log m + \log k$, the benchmark allocations $\{\vec{B}^{(t)}\}_{t=1}^{t=\gamma}$ and pricing solutions $\{(\vec{p}^{(t)},\vec{S}^{(t)})\}_{t=1}^{t=\gamma}$ generated by the above algorithm satisfy the generalized charging property (Definition \ref{defn_framework2}) with $\alpha=\beta=2\gamma$, and thus for some $\ell$ we have that
$$SW(\vec{A})\leq 4(\log m + \log k)^2 Rev(\vec{p}^{(\ell)}, \vec{S}^{(\ell)}).$$
\label{thm.XoScharging}
\end{theorem}

Note that although this theorem immediately implies that at least one solution $(\vec{p}^{(\ell)},\vec{S}^{(\ell)})$ has high revenue compared to the welfare of the starting allocation $\vec{A}$, this does not yet give us the result we need since the allocation $\vec{S}^{(\ell)}$ is specific to one particular arrival order of buyers, whereas the actual arrival order is either random or adversarial.

\subsubsection*{Proof of Theorem \ref{thm.XoScharging}}
Fix $\gamma = \log m + \log k$. By Claim \ref{clm_frame2rev}, it suffices if we show that for every $1 \leq t \leq \gamma - 1$,

\begin{equation}
\label{eqn_charging1}
SW(\vec{B}^{(t)}) - SW(\vec{B}^{(t+1)}) \leq 2\gamma Rev(\vec{p}^{(t)},\vec{S}^{(t)}) + \frac{SW(\vec{B}^{(t)})}{2\gamma}
\end{equation}

and that, $$SW(\vec{B}^{(\gamma)}) \leq 2\gamma Rev(\vec{p}^{(\gamma)},\vec{S}^{(\gamma)}).$$

Before actually proving these claims, we provide an outline of the proof. \begin{enumerate}
\item Consider some iteration $t$ of our core algorithm where the benchmark solution is $\vec{B}^{(t)}$. As mentioned in the algorithm, let $\vec{S}^{(t)}$ denote the allocation obtained upon running the sequential posted pricing mechanism with prices that are a scaled down version of each good's average utility in $\vec{B}^{(t)}$, (as in $Prices(\vec{B}^{(t)},\gamma)$). Then, we claim and prove that $SW(\vec{B}^{(t)}) - SW(\vec{S}^{(t)})$ can be bounded in terms of the mechanism's revenue and the welfare of $\vec{B}^{(t)}$, analogous to Definition~\ref{defn_framework2}. This has obvious implications in showing that our solutions satisfy Equation~\ref{eqn_charging1} for a given value of $t$ as long as the core algorithm chooses $\vec{B}^{(t+1)}$ to be the output of the sequential mechanism (Case I).

\item Given a benchmark allocation vector $\vec{B}^{(t)}$, once again consider running the sequential mechanism with scaled down prices, and let $\vec{T}$ denote a sub-allocation of $\vec{B}^{(t)}$ obtained by allocating the `unsold' items back to the original buyers (as in Alloc-Unsold$(\vec{B}^{(t)},\vec{S}^{(t)})$). Then, $SW(\vec{B}^{(t)}) - SW(\vec{T})$ can be bounded purely in terms of the revenue of the above sequential mechanism. Therefore, our solutions satisfy Equation~\ref{eqn_charging1} in the second case where $\vec{B}^{(t+1)}$ is defined to be the output of Alloc-Unsold.

\item In the final act, we prove that for each $t$, $\vec{B}^{(t+1)}$ allocates at most half the total number of items as $\vec{B}^{(t)}$, and so at most two items are allocated in $\vec{B}^{(\gamma)}$. In this case, the tail part of our core algorithm, where exactly one good is sold, results in a revenue that is only a constant factor away from the social welfare of $\vec{B}^{(\gamma)}$.

\end{enumerate}

\subsubsection*{General Claims}

We begin with our first, somewhat general claim that connects the social welfare of the benchmark allocation to the revenue and welfare of the solution returned by the sequential posted pricing mechanism for a certain type of `scaled-down' pricing scheme. As we will show later, this can be used to bound the difference in welfare of any two consecutive benchmark solutions (in our core algorithm) as long as the second solution is defined to be the output of the sequential mechanism. Owing to their rather general nature, the proofs of the following two claims are presented in Appendix~\ref{app:sec4} in order to not distract the reader from the main ideas underlying this theorem.

\begin{claim}
\label{clm_case1gen}
Given any allocation $\vec{B}$ and some parameter $\alpha \geq 1$, consider the sequential posted pricing mechanism with arbitrary buyer arrival order for supply constraint $q_i = k_i(\vec{B})$, and price $p_i = \frac{U_i(\vec{B})}{\alpha q_i}$ for all $i \in \mathcal{I}$. Suppose that $\vec{S}$ denotes the output allocation of this mechanism. Then,

$$SW(\vec{B}) - Surp(\vec{p},\vec{S}) \leq \alpha Rev(\vec{p},\vec{S}) + \frac{SW(\vec{B})}{\alpha}.$$
\end{claim}

Recall that $Surp(\vec{p},\vec{S})$ is the total surplus or utility derived by the buyers under the given prices. Since, we know that social welfare equals the sum of revenue and surplus, we can use the above claim to obtain an upper bound on $SW(\vec{B}) - SW(\vec{S})$, which we state as a corollary after proving the claim.

\begin{corollary}
\label{corr_gencase1}
Given any allocation $\vec{B}$ and some parameter $\alpha \geq 1$, consider the sequential posted pricing mechanism with arbitrary buyer arrival order for supply constraint $q_i = k_i(\vec{B})$, and price $p_i = \frac{U_i(\vec{B})}{\alpha q_i}$ for all $i \in \mathcal{I}$. Suppose that $\vec{S}$ denotes the output allocation of this mechanism. Then,

$$SW(\vec{B}) - SW(\vec{S}) \leq \alpha Rev(\vec{p},\vec{S}) + \frac{SW(\vec{B})}{\alpha}.$$
\end{corollary}

Our next series of lemmas are useful in showing that Equation~\ref{eqn_charging1} holds in the (second) case when our core algorithm uses the function Alloc-Unsold$(\vec{B}^{(t)},\vec{S}^{(t)})$ to select the next benchmark allocation. We begin with an obvious claim on the average of $n$ numbers that we state here without proof.

\begin{lemma}
\label{lem_propertyaverage}
Consider a sequence of non-negative real numbers $v_1 \geq v_2 \geq \ldots \geq v_n$ and let $V = \sum_{i=1}^n v_i$. Then, for any $0 \leq r \leq n$,
we have that $\frac{V}{n} r + \sum_{i=1}^{n-r}v_i \geq V$.
\end{lemma}

\begin{claim}
\label{clm_allocunsoldgen}
Given any allocation $\vec{B}$ and some parameter $\alpha \geq 1$, consider the sequential posted pricing mechanism with arbitrary buyer arrival order for supply constraint $q_i = k_i(\vec{B})$, and price $p_i = \frac{U_i(\vec{B})}{\alpha q_i}$ for all $i \in \mathcal{I}$. Suppose that $\vec{S}^M$ denotes the output allocation of this mechanism, and $\vec{S}$ is the allocation returned by $\text{Alloc-Unsold}(\vec{B},\vec{S}^M)$. Then,

$$SW(\vec{B}) - SW(\vec{S}) \leq \alpha Rev(\vec{p},\vec{S}^M).$$
\end{claim}

\subsubsection*{Specific Lemmas Pertaining to Core Algorithm}

Applying Claims~\ref{clm_case1gen} and~\ref{clm_allocunsoldgen} in the context of our core algorithm, we will later prove that $SW(\vec{B}^{(t)}) - SW(\vec{B}^{(t+1)}) \leq 2\gamma Rev(\vec{p}^{(t)},\vec{S}^{(t)}) + \frac{SW(\vec{B}^{(t)})}{2\gamma}$ for all $t \leq \gamma - 1$. In order to show that the charging property is obeyed, it only remains for us to show that $SW(\vec{B}^{(\gamma)}) \leq \alpha Rev(\vec{p}^{(\gamma)},\vec{S}^{(\gamma)})$. We will begin with the lemma that highlights the `item halving' nature of our algorithm. 

\begin{lemma}
\label{lem_itemhalving}
During the course of our algorithm, the total number of items allocated in successive benchmark solutions is at least halved, i.e., for any $t \leq \gamma - 1$,

$$\theta(\vec{B}^{(t)}) \geq 2\theta(\vec{B}^{(t+1)}).$$
\end{lemma}
\begin{proof}
The proof is not hard to see, and we show this in two cases. Recall that at the end of round $t$ of our core algorithm, the choice of $\vec{B}^{(t+1)}$ depends on whether or not $\theta(\vec{B}^{(t)}) \geq 2\theta(\vec{S}^{(t)})$. If the condition holds, $\vec{B}^{(t+1)} := \vec{S}^{(t)}$ and the lemma follows directly.

Suppose that this is not the case, and that $\theta(\vec{B}^{(t)}) < 2\theta(\vec{S}^{(t)})$. In this case, our algorithm chooses $\vec{B}^{(t+1)}$ to be the output of $\text{Alloc-Unsold}(\vec{B}^{(t)}, \vec{S}^{(t)})$. We know that for the given choice of inputs, the function Alloc-Unsold allocates exactly $k_i(\vec{B}^{(t)}) - k_i(\vec{S}^{(t)})$ copies of good $i \in \mathcal{I}$ to the buyers. Summing this up over all the goods, we get that $\theta(\vec{B}^{(t+1)}) + \theta(\vec{S}^{(t)}) = \theta(\vec{B}^{(t)})$ in the case that $\theta(\vec{B}^{(t)}) < 2\theta(\vec{S}^{(t)})$. And so, we get that $2\theta(\vec{B}^{(t+1)}) < \theta(\vec{B}^{(t)}).$
\end{proof}

We know that in the initial allocation to our core algorithm $\vec{A}$ (and hence $\vec{B}^{(1)}$), at most $mk$ copies of goods are allocated since $k$ denotes the average supply. In every successive round, the number of items in the (next) benchmark solution is at least halved. Therefore, after $\log(mk) - 1$ rounds, i.e., in $\vec{B}^{(\gamma)}$, at most $2$ items are allocated to the buyers.

\begin{corollary}
The total number of items allocated in the solution $\vec{B}^{(\gamma)}$ is at most two.
\end{corollary}

\begin{lemma}
\label{lem_swgammaupper}
Define $(i^*,j^*) := \argmax_{i \in \mathcal{I}, j \in \mathcal{N}}v_j(i)$. The social welfare of the allocation $\vec{B}^{(\gamma)}$ is at most $2v_{j^*}(i^*)$.
\end{lemma}
\begin{proof}
Suppose that $i_1$ and $i_2$ are the identities of the two (not necessarily distinct) goods allocated to the buyers in $\vec{B}^{(\gamma)}$, and $j_1, j_2$ refer to the (not necessarily distinct) buyers receiving these goods. Reducing the social welfare in terms of the XoS clauses, we have that

$$SW(\vec{B}^{(\gamma)}) = x^{\vec{B}^{(\gamma)}_{j_1}}_{j_1}(i_1) + x^{\vec{B}^{(\gamma)}_{j_2}}_{j_2}(i_2).$$

However, it is not particularly hard to reason that for any buyer $j$, item $i$, and XoS clause $x_j$ corresponding to this buyer, $x_j(i) \leq v_{j^*}(i^*).$ Thus, $x^{\vec{B}^{(\gamma)}_{j_1}}_{j_1}(i_1) + x^{\vec{B}^{(\gamma)}_{j_2}}_{j_2}(i_2) \leq v_{j^*}(i^*) + v_{j^*}(i^*)$. The lemma follows.
\end{proof}


Thus, since the revenue obtained by $(\vec{p}^{(\gamma)},\vec{S}^{(\gamma)})$ is $v_{j^*}(i^*)-\epsilon$, we know that for a small enough epsilon it must be that $SW(\vec{B}^{(\gamma)}) \leq 2\gamma Rev(\vec{p}^{(\gamma)},\vec{S}^{(\gamma)})$, as desired.

\subsubsection*{Wrapping up the Proof}
Equipped with the various pieces, we can now complete the proof of our black-box reduction by showing that our solutions obey the conditions outlined in Equation~\ref{eqn_charging1} and hence Definition~\ref{defn_framework2}. Fix the input to the core algorithm, i.e., allocation $\vec{A}$, 
and $\gamma$ as defined previously. Consider any iteration $t \leq \gamma-1$ of our algorithm. Let us proceed in two cases depending on the algorithm's choice of $\vec{B}^{(t+1)}$.

In the first case, suppose that our algorithm identifies $\theta(\vec{B}^{(t)}) \geq 2\theta(\vec{S}^{(t)})$ and defines $\vec{B}^{(t+1)} = \vec{S}^{(t)}$. Then, we apply Corollary~\ref{corr_gencase1} with $\vec{B} = \vec{B}^{(t)}$, $\vec{S} = \vec{S}^{(t)} = \vec{B}^{(t+1)} $, and $\alpha = 2\gamma$ and get that

$$SW(\vec{B}^{(t)}) -SW(\vec{B}^{(t+1)}) \leq 2\gamma Rev(\vec{p}^{(t)},\vec{S}^{(t)}) + \frac{SW(\vec{B}^{(t)})}{2\gamma}.$$

In the second case, we have that $\theta(\vec{B}^{(t)}) < 2\theta(\vec{S}^{(t)})$ and so, we set $\vec{B}^{(t+1)} = \text{Alloc-Unsold}(\vec{B}^{(t)}, \vec{S}^{(t)})$. For this case, we refer to Claim~\ref{clm_allocunsoldgen} with $\vec{B} = \vec{B}^{(t)}, \vec{S}^M = \vec{S}^{(t)}$, and $\alpha = 2\gamma$ to obtain

$$SW(\vec{B}^{(t)}) -SW(\vec{B}^{(t+1)}) \leq 2\gamma Rev(\vec{p}^{(t)},\vec{S}^{(t)}) \leq 2\gamma Rev(\vec{p}^{(t)},\vec{S}^{(t)}) + \frac{SW(\vec{B}^{(t)})}{2\gamma}.$$

Finally, the last condition that $SW(\vec{B}^{(\gamma)}) \leq 2 Rev(\vec{p}^{(\gamma)},\vec{S}^{(\gamma)})$ was proven above. Therefore, our solutions satisfy Definition~\ref{defn_framework2} and as a consequence, we can apply Claim~\ref{clm_frame2rev} with $\alpha = \beta = 2\gamma$. According to the claim, there exists some $\ell \in [1,\gamma]$ such that $Rev(\vec{p}^{(\ell)},\vec{S}^{(\ell)}) \geq \frac{SW(\vec{A})}{2(\log m + \log k)^2}(1 - \frac{\gamma - 1}{2\gamma}) \geq \frac{SW(\vec{A})}{4(\log m + \log k)^2}.$ This completes the proof. \hfill $\qed$

\subsection{Applying the Item Halving Framework to Form Pricing with High Revenue}\label{sec:applying}

We are now in a position to leverage Theorem~\ref{thm.XoScharging} as a black-box mechanism to obtain computationally efficient posted price mechanisms with good revenue guarantees. Specifically, for instances having XoS valuations, we can use the $\frac{e}{e-1}$-approximation algorithm for the allocation  problem~\cite{dobzinskiNS10} to form an initial allocation $\vec{A}$ with high welfare. The item-halving algorithm gives us a pricing and allocation $(p,S)$ whose revenue is a $O((\log m + \log k)^2)$-approximation to the optimum social welfare, and therefore to the revenue obtained by any other pricing mechanism. However, $S$ is a specific allocation and only arises for some buyer arrival orders. By choosing the function $\Gamma$ appropriately, we can make sure that the revenue resulting from all other arrival orders is no worse than that of $S$.

\subsubsection{Adversarial Buyer Arrival}

\begin{theorem}
\label{thm_computationalxos}
\textbf{(Main Computational Result)} Given a set of arrival orders $\Pi$, we can compute in time polynomial in $N$, $m$, and $|\Pi|$ a set of prices $\vec{p}$ and supply constraints $\vec{q}$ such that the revenue guaranteed by the sequential posted pricing mechanism with these parameters for all $\pi\in\Pi$ is a $\frac{4e}{e-1}(\log m + \log k)^2$-approximation to the optimum social welfare, and thus to the optimum revenue as well.
\end{theorem}

For the case where the buyer arrival order is adversarial with $|\Pi|$ being of polynomial size, Theorem \ref{thm_computationalxos} immediately gives us an efficient algorithm for forming a sequential pricing mechanism which achieves a $O((\log m+\log k)^2)$ approximation for maximum revenue. For the case where $\Pi$ consists of all possible arrival orders (truly adversarial), this still gives us an interesting existence result: for every instance of XOS valuations, there exist {\em fixed} prices and supply constraints such that, no matter what order the buyers arrive, the revenue of the outcome is a $O((\log m + \log k)^2)$ approximation to the maximum possible welfare. 

\begin{proof}
Set $\vec{A}$ to be the allocation which yields a $\frac{e}{e-1}$-approximation to maximum social welfare; this can be obtained from \cite{dobzinskiNS10}. Set $\Gamma(\vec{p},\vec{B})$ to be as follows. For each $\pi\in \Pi$, run the sequential posted price mechanism with prices $\vec{p}$ and supply constraints $\vec{q}=k(\vec{B})$, and choose $\pi$ to be the arrival order that results in the smallest revenue. Set $\Gamma(\vec{p},\vec{B})$ to be this arrival order $\pi$.

Due to Theorem \ref{thm.XoScharging}, using the Item-Halving framework with this function $\Gamma$ and starting allocation $\vec{A}$ gives us prices and allocation $(\vec{p}^{(\ell)},\vec{S}^{(\ell)})$ so that $SW(\vec{A})\leq 4(\log m + \log k)^2 Rev(\vec{p}^{(\ell)}, \vec{S}^{(\ell)})$. However, note that $\vec{S}^{(\ell)}$ is the allocation that was obtained by running the sequential mechanism with prices $\vec{p}^{(\ell)}$, supply constraints $k(\vec{B}^{(\ell)})$, and arrival order $\pi=\Gamma(\vec{p}^{(\ell)},\vec{B}^{(\ell)})$; thus by our choice of $\Gamma$ any other arrival order with prices $\vec{p}^{(\ell)}$ and supply $k(\vec{B}^{(\ell)})$ should only result in better revenue than $Rev(\vec{p}^{(\ell)}, \vec{S}^{(\ell)})$. Therefore, running the sequential posted price mechanism with parameters $\vec{p}^{(\ell)}$ and $k(\vec{B}^{(\ell)})$ will always result in revenue which is at most a factor $4(\log m + \log k)^2$ away from $SW(\vec{A})$, no matter what arrival order from $\Pi$ is chosen by the adversary. Since $SW(\vec{A})$ is a $\frac{e}{e-1}$ approximation to the maximum social welfare, and since the maximum social welfare is always at least the maximum possible revenue, this gives us the desired result.

We must still argue that the above holds when $\ell=\gamma$, since the last allocation $\vec{S}^{(\ell)}$ is generated using the special tail condition, instead of using $\Gamma$. Notice, however, that with prices as defined in Figure \ref{fig:tail}, and supply constraints $q_{i^*}=1$, $q_{i}=0$ for all $i\neq i^*$, all arrival orders will result in the same revenue, and thus if $Rev(\vec{p}^{(\gamma)}, \vec{S}^{(\gamma)})$ is a good approximation to maximum welfare, then so is the revenue achieved by any arrival order with prices $\vec{p}^{(\gamma)}$ and supply constraints $\vec{q}$.
\end{proof}

\subsubsection{Random Arrival Orders}
We now consider a more realistic model~\cite{kesselheimKN15} of buyer arrival that is determined by an arbitrary distribution $D^{\Pi}$ over the set of arrival orders. We make the standard assumption that we have access to an oracle from which we can draw samples corresponding to the distribution $D^{\Pi}$. Building on the extensive machinery developed in the previous sections, we design a mechanism that with high probability achieves a $O((\log m + \log k)^2)$-approximation for revenue, which is the same asymptotic factor that was obtained for the partially adversarial model. On the surface, it may seem surprising that our approximation factors are completely independent of the nature of the distribution and unlike some other works, our algorithm is not finely tuned to the properties underlying the distribution. At a high level, the generality of our result for the random arrival model comes from the fact that, as argued above, there always exist prices that obtain the same poly-logarithmic approximation factors even for the more general fully adversarial arrival model, even if these prices are difficult to compute. In some senses, what we show in this section is that with high probability, we are able to compute prices that work well for fully adversarial arrivals or at least a significant portion of the set of all buyer arrival orders.

\begin{theorem}
\label{thm_xosrandom}
Given oracle access to a distribution $D^{\Pi}$ over the buyer arrival orders, we can design in poly-time a sequential posted pricing mechanism $\mathcal{M}$ such that with high probability,

$$SW(OPT) \leq \frac{8e}{e-1}(\log m + \log k)^2 E_{\pi \sim D^{\Pi}}[Rev(\mathcal{M,\pi})],$$

where $Rev(\mathcal{M,\pi})$ is the revenue achieved by the mechanism $\mathcal{M}$ when the buyer arrival order is given by $\pi$.
\end{theorem}

We remark here that while the revenue achieved by our mechanism is in expectation over the arrival orders, the `high probability' refers to the randomization of the algorithm itself. The core algorithm that achieves the desired approximation guarantee is the same as before; the only difference is in the function $\Gamma$, which is shown in figure \ref{fig.randomGamma}.


\begin{figure}[h]
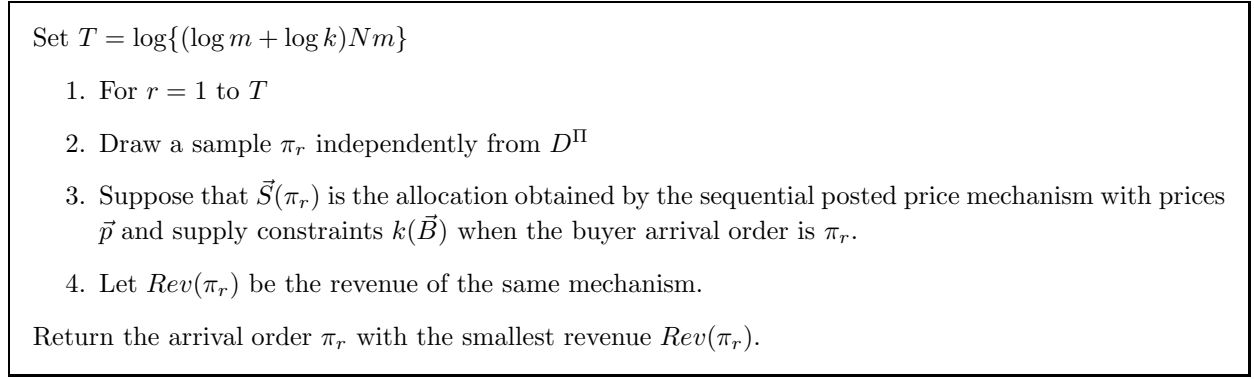

\begin{framed}
Set $T = \log\{(\log m + \log k)Nm\}$
\begin{enumerate}
\item For $r = 1$ to $T$
\item Draw a sample $\pi_r$ independently from $D^{\Pi}$
\item Suppose that $\vec{S}(\pi_r)$ is the allocation obtained by the sequential posted price mechanism with prices $\vec{p}$ and supply constraints $k(\vec{B})$ when the buyer arrival order is $\pi_r$.
\item Let $Rev(\pi_r)$ be the revenue of the same mechanism.
\end{enumerate}
Return the arrival order $\pi_r$ with the smallest revenue $Rev(\pi_r)$.
\end{framed}
\caption{Function $\Gamma(\vec{p}, \vec{B})$ specifically for random arrival orders}
\label{fig.randomGamma}
\end{figure}

Instead of iterating over a finite set of arrival orders $\Pi$ as we did previously, the new sequential mechanism function draws a polynomial number of arrival orders according to the input distribution, and for each of these orders simulates the sequential posted pricing mechanism with the same set of prices and supply constraints. Finally, the function returns the outcome of the posted pricing mechanism (corresponding to some arrival order), which has the worst revenue. 

\begin{proof}
%
By once again using the initial allocation $\vec{A}$ to be the same as for the adversarial arrival case, we know that the Item-Halving Framework once again results in a solution $(\vec{p}^{(\ell)},\vec{S}^{(\ell)})$ so that $Rev(\vec{p}^{(\ell)}, \vec{S}^{(\ell)})$ is a $\frac{4e}{e-1}(\log m + \log k)^2$ approximation to the optimum welfare.

Of course, it is not sufficient to just prove that $Rev(\vec{p}^{(\ell)}, \vec{S}^{(\ell)})$ is a good approximation to the optimum welfare. We actually need to show that for our choice of price and supply, i.e., $\vec{p}^{(\ell)}, k(\vec{B}^{(\ell)})$, the mechanism $\mathcal{M}$ with these parameters results in the following guarantee: $SW(\vec{A}) \leq 8(\log m + \log k)^2 E_{\pi \sim D^{\Pi}}[Rev(\mathcal{M,\pi})].$ Towards this end, we define $\overline{Rev}^{(t)}$ to be the expected revenue of the sequential mechanism with parameters $\vec{p}^{(t)}, k(\vec{B}^{(t)})$, where the expectation is taken over the entire distribution $D^{\Pi}$. Note that ideally, we would like to identify the index that maximizes $\overline{Rev}^{(t)}$; however, we do not have access to this statistic and therefore, will settle for `a strong estimator' $Rev(\vec{p}^{(\ell)}, \vec{S}^{(\ell)})$, which is actually the minimum revenue over several samples.

\begin{lemma}
\label{lem_gensamples}
Suppose that $x_1, x_2, \ldots, x_n$ are $n$ samples that are drawn i.i.d from some distribution $D$. Then for any arbitrary function $f$,

$$Pr\{\min_{i=1}^n f(x_i) > 2 E_{x \sim D}[f(x)]\} \leq \frac{1}{2^n}.$$
\end{lemma}
\begin{proof}
By Markov's inequality, we know that for any $1 \leq i \leq n$, $Pr(f(x_i) \geq 2E[f(x)]) \leq \frac{1}{2}$. Therefore,
the total probability is at most $\frac{1}{2^n}$.
\end{proof}

Thus, we actually need to show that $SW(\vec{A}) \leq O(\log m + \log k)^2 \overline{Rev}^{(\ell)}$. We claim that this holds with high probability. Specifically, if $\vec{p}^{(\ell)}, k(\vec{B}^{(\ell)})$ denote the prices and supply constraints computed by our core algorithm, then with high probability, the expected revenue generated by the sequential posted pricing mechanism with these parameters is within a $O(\gamma^2)$-factor of the social welfare of $\vec{A}$.

Consider $\overline{Rev}^{(\ell)}$ and how well is this statistic estimated by $Rev^{(\ell)}=Rev(\vec{p}^{(\ell)}, \vec{S}^{(\ell)})$. We claim that with high probability $\overline{Rev}^{(\ell)} \geq \frac{Rev^{(\ell)}}{2}$. Consider some iteration $t$, applying Lemma~\ref{lem_gensamples} with $n=T$ and the function $f$ denoting the revenue of the mechanism, we get the following bound on the minimum revenue taken over $T$ sample arrival orders: $Pr\{Rev^{(t)} > 2 \overline{Rev}^{(t)}\} \leq \frac{1}{2^T} = \frac{1}{Nm \log(mk)}.$

Therefore, we can apply the union bound to finish off the proof,

\begin{align*}
Pr\{\overline{Rev}^{(\ell)} & < \frac{Rev^{(\ell)}}{2} \} \leq Pr\{\bigcup_{t=1}^{\gamma} \overline{Rev}^{(t)} < \frac{Rev^{(t)}}{2} \}\\
& \leq \sum_{t=1}^{\gamma} Pr\{\overline{Rev}^{(t)}  < \frac{Rev^{(t)}}{2} \}  \leq \sum_{t=1}^{\gamma} \frac{1}{Nm \log(mk)} = \frac{1}{Nm}.
\end{align*}

So, with probability $(1 - \frac{1}{Nm})$, $\overline{Rev^{(\ell)}}$ is a $8(\log m + \log k)^2$ approximation to the social welfare of $\vec{A}$, which in turn is a $\frac{e}{e-1}$-approximation to the optimum welfare, if one takes $\vec{A}$ to be the allocation output by the algorithm in~\cite{dobzinskiNS10}.
\end{proof}


\subsubsection{Bicriteria Approximations}
In Section \ref{sec:pricedoubling}, we leveraged the simple charging property to design bi-criteria algorithms that simultaneously maximize revenue and welfare, culminating in mechanisms that guarantee a constant factor of the optimum social welfare without sacrificing much revenue. Can we provide similar guarantees here? The level of generality at which our second framework operates precludes unconditional constant factor approximations for welfare if we also require good revenue approximations. That said, we are still able to obtain interesting bicriteria results for our sequential mechanisms when buyers have XoS valuations, with the exact revenue-welfare trade-off depending on the instance. At a high level, our bicriteria result can be summed up as follows: either our mechanism obtains a constant factor social welfare in addition to the log-squared revenue guarantees provided by Theorem \ref{thm_computationalxos} or both the revenue and welfare are a logarithmic approximation to the optimum welfare (as opposed to log-squared).

\begin{theorem}
\label{thm_bicriteriaxos}
Given a polynomially bounded set of arrival orders $\Pi$, we can compute in poly-time a set of prices and supply constraints such that the sequential posted pricing mechanism with these parameters provides one of the following two guarantees:
\begin{enumerate}
\item The mechanism's revenue is a $\Theta(\log m + \log k)^2$-approximation to the optimum welfare and its social welfare is a $O(1)$-approximation to the same objective.\\
\textbf{OR}
\item Both the social welfare and revenue of our mechanism are a $O(\log m + \log k)$-factor smaller than the optimum social welfare.
\end{enumerate}

\end{theorem}

Depending on the instance, our approximations will fall either in the first realm (constant factor welfare) or the second (log factor revenue). Even if we wish to ignore the dichotomy, the social welfare guaranteed by our mechanism is at most a $O(\log m + \log k)$-factor away from the optimum welfare in the worst case. The exact algorithm that achieves the above bicriteria approximation is as follows: \emph{Run the core algorithm with $(i)$ An allocation $\vec{A}$ that is a $(1-\frac{1}{e})$-approximation to the optimum welfare~\cite{dobzinskiNS10}, $(ii)$ Set of arrival orders $\Pi$, and $(iii)$ $\gamma = \log m + \log k$. Use the same $\Gamma$ function as the one for Theorem \ref{thm_computationalxos}, i.e., the one that returns the arrival order yielding the worst possible revenue. Let $\ell$ be the smallest index such that $Rev(\vec{p}^{(\ell)},\vec{S}^{(\ell)}) \geq \frac{1}{8(\log m + \log k)^2}SW(\vec{A})$. The sequential mechanism corresponding to $\vec{p}^{(\ell)}, \vec{q}^{(\ell)}=k(\vec{B}^{(\ell)})$ achieves the desired bicriteria guarantee. }

The actual proof of the theorem is deferred to the appendix.

\section{Item Halving Framework applied to Multi-unit Markets}\label{sec:multi-item}
We now highlight the versatility of our item halving framework by showing that it leads to good approximations for revenue even when buyer valuations may not belong to the class of XoS functions. Specifically, we consider the popular multi-unit market setting~\cite{branzeiFMZ16,feldmanFLS12} where buyers only care about the number of items they receive, and show that a small variant of our previously developed algorithm yields a $O(\log m)$-approximation for revenue. Previously ~\cite{balcanBM08} provided a sequential mechanism with a $O(\log m + \log(N))$ guarantee for revenue; not only do we improve upon the approximation factor by removing the dependence on $N$, but we provide a {\em simultaneous} mechanism with the given approximation factor. One can therefore infer that it is possible to achieve good guarantees in this setting without giving up on envy-freeness. Finally, our result also implies that the item halving framework is not specifically tied to just one type of mechanism (sequential) or a specific class of valuations (XoS).

A multi-unit market consists of a single good and each buyer's valuation depends only on the number of copies of this good that she receives. For the purpose of notational consistency, we assume that the market consists of a set $\mathcal{I}$ of $m$ goods with one copy each and every buyer's valuation function depends only on the cardinality of the set of goods, i.e., for any $j \in \mathcal{N}$, and any $S, T \subseteq \mathcal{I}$, $v_j(T) = v_j(S)$ if $|S| = |T|$. 
One can interpret this setting as a market comprising of a single good with $m$ copies therein. Note that such valuations can also exhibit complementarities, and therefore do not fall under the XoS class. Additionally, we make the fairly standard assumption (e.g., see~\cite{balcanBM08}) that no single buyer wishes to buy out more than half of the entire supply of goods in the market:

\begin{definition}{(No Overwhelming Buyer Assumption)} We say that a given instance with multi-unit valuations satisfies the {\em no overwhelming buyer assumption} if for every buyer $j$, $v_j(\frac{m}{2} + q) = v_j(\frac{m}{2})$ for any non-negative integer $q$.
\end{definition}

Based on the assumption, we will only consider allocations where no single buyer consumes more than $\frac{m}{2}$ copies of the good.

\begin{theorem}\label{thm:multi-unit}
Given any instance of a multi-unit market that satisfies the no overwhelming buyer assumption and an input allocation $\vec{A}$, we can compute in poly-time a single price $\tilde{p}$ such that the simultaneous mechanism with price $\tilde{p}$ on all copies of the good achieves a $O(\log m)$-revenue approximation to $SW(\vec{A})$ as long as $m \geq 2$.
\end{theorem}
Notice that in multi-unit settings, it is essential for a non-discriminatory mechanism to post a single price on all of the goods since these are identical from the buyers' perspective. One can easily couple the above theorem with any of the known algorithms for computing welfare maximizing allocations~\cite{dobzinskiN10} to obtain a simultaneous mechanism whose revenue is a $O(\log m)$-approximation to the optimum welfare, and hence revenue. The algorithm that achieves this approximation guarantee is based on the same item halving core algorithm as the XoS case and therefore, the proof bears overarching similarities to our previous proofs, and therefore, is presented in Appendix~\ref{app:sec5}. However, a few constraints that arise for the multi-unit case are quite different compared to the XoS setting and we carefully handle these in our algorithm.

\section{Conclusion and Open Questions}
The contributions made by this work can be evaluated at two levels. The first and more fundamental contribution is an intuitive but general framework for revenue maximization based on the commonly used \emph{price doubling} approach, and our new technique of \emph{item halving}. The price doubling framework captures the essence of a number of revenue maximization algorithms, previously developed in an ad-hoc fashion for specific settings. We unify many of these results under a common banner, and using our framework as a sledgehammer, transform all of these results into bicriteria approximation algorithms with only a constant factor loss in revenue. We believe that the design of simple frameworks that unify existing results, and then go on to improve upon them, is a highly important research agenda in computer science. 
We expect that the framework developed in this paper, especially the new insights obtained via item halving, will be useful for revenue maximization in other settings.

Our second contribution is more technical: we present the first known algorithm with a log-squared approximation factor for revenue maximization with XoS buyers. In doing so, we partially resolve a decade long open question~\cite{balcanBM08,chakrabortyHK09} about the existence of polylogarithmic approximations for revenue maximization based on static item pricing. Our result is surprisingly powerful: it implies that even for quite general combinatorial valuations (which are known to have hidden complementarities~\cite{feldmanGL16}), we can use arguably the simplest pricing scheme known to us (one number per good) and approximate an all-encompassing benchmark, that is the optimum social welfare. 
We do however remark that the previously known guarantees of \cite{balcanBM08} hold for the more general class of subadditive valuations, and it would be interesting to see if our results extend in that direction with an additional logarithmic loss, since subadditive functions are known to be log-approximate XoS functions.

Our results are shown for a full information model, which appears to be in some senses necessary, since our algorithm relies crucially on computing a series of allocations based on simulating the mechanism. Of course, a natural question to ask is whether our results extend to a more realistic Bayesian information setting, where the seller is only aware of the distribution from which buyers' valuations are drawn. In moving to the Bayesian framework, we immediately run into difficulties. For instance, even for single-item Bayesian auctions, it is known that no posted price obtains a reasonable revenue in comparison to the optimum welfare~\cite{alaeiHNPY15}. One natural way to circumvent this barrier would be to develop benchmarks that are weaker than the optimum welfare as in~\cite{caiDW16,chawlaM16}; alternatively, we could assume some natural properties on the distributions themselves such as monotone hazard rate. Extending our framework and techniques to Bayesian and other partial information settings is perhaps the most clear direction for future work.

\bibliographystyle{plain}

\bibliography{bibliography}

\appendix
\newpage

\section{Formal Definitions of Buyer Valuations}
\label{app:definitions}
In this section, we formally define all the classes of valuation functions that we refer to in our work. Recall that our model comprises of monotone combinatorial valuation functions $v:2^{\mathcal{\mathcal{I}}} \to \mathbb{R}^{+} \cup \{0\}$. Finally, given a valuation $v$ and prices $\vec{p}$, we say that a set $S \subseteq \mathcal{I}$ belongs to the demand set of $v$ with respect to $\vec{p}$ if $S$ is a utility-maximizing bundle under these prices, i.e., $S = \argmax_{T \subseteq \mathcal{I}}(v(T) - \sum_{i \in T}p_i)$. 
\begin{description}

\item [Unit-Demand] In a unit-demand valuation, there exists a set of values $(v_i)_{i \in \mathcal{I}}$ such that for every $T \subseteq \mathcal{I}$, $v(T) = \max_{i \in T}v_i$.

\item [Gross Substitutes] A valuation satisfies the gross substitutes property if for every two vectors of prices $\vec{p_1}$, $\vec{p_2}$ such that $\vec{p_2} \geq \vec{p_1}$, and every bundle $S$ in the demand set of $v$ given $\vec{p_1}$, there exists a bundle $T$ in the demand set of $v$ with respect to $\vec{p_2}$ which contains every good $i \in S$ whose price in $\vec{p_2}$ is unchanged from its price in $\vec{p_1}$. 

\item [Submodular] For any two sets $S, T$ with $T \subseteq S \subseteq \mathcal{I}$, and any item $i \in \mathcal{I}$, $v(S \cup \left\lbrace i \right\rbrace) - v(S) \leq v(T \cup \left\lbrace i \right\rbrace) - v(T)$.

\item [XoS or Fractionally Subadditive] $\exists$ a set of additive functions $(a_1, \ldots, a_r)$ such that for any $T \subseteq \mathcal{N}$, $v(T) = \max_{j=1}^{r} a_j(T)$. These additive functions are referred to as {\em clauses.} Recall that an additive function $a_j$ has a single value $a_j(i)$ for each $i \in \mathcal{N}$ so that for a set $T$ of agents, $a_j(T) = \sum_{i \in T}a_j(i)$.

\item [Subadditive] A valuation is said to be subadditive if any two sets $S,T \subseteq \mathcal{I}$, $v(S \cup T) \leq v(S) + v(T)$.

\item [Multi-Unit] For any two sets $S,T \subseteq \mathcal{I}$ such that $|S| = |T|$, we have $v(S) = v(T)$.

\item [Single-Minded] A valuation is said to be single-minded if there exists a single set $T$ and value $x$ such that $v(S) = x$ if and only if $T \subseteq S$ and $v(S) = 0$ otherwise.
\end{description}

The valuation functions defined above follow a nice hierarchy of complement-free valuations, namely: Unit-Demand $\subseteq$ Gross Substitutes $\subseteq$ Submodular $\subseteq$ XoS $\subseteq$ Subadditive. On the other single-minded valuations represent among the simplest class of valuations where buyers exhibit complements.

\section{Proofs from Section \ref{sec:model}}

\subsubsection*{Proof of Claim~\ref{clm_simplecharging}}

%
%
\begin{proof}
Suppose that $\ell$ is the index that maximizes revenue, i.e., for all $1 \leq t \leq \gamma$, $Rev(\vec{p}^{(\ell)}, \vec{S}^{(\ell)}) \geq Rev(\vec{p}^{(t)}, \vec{S}^{(t)})$. Performing the telescoping summation on the following system of inequalities $SW(\vec{S}^{(t)}) - SW(\vec{S}^{(t+1)}) \leq \alpha Rev(\vec{p}^{(t)}, \vec{S}^{(t)}),$ from $t=1$ to $t=\gamma-1$ and adding the inequality $SW(\vec{S}^{(\gamma)}) \leq \alpha Rev(\vec{p}^{(\gamma)}, \vec{S}^{(\gamma)})$ to the result of the summation, we get:

$$SW(\vec{S^{(1)}}) \leq \alpha \sum_{t=1}^{\gamma} Rev(\vec{p}^{(t)}, \vec{S}^{(t)}) \leq \alpha \sum_{t=1}^{\gamma} Rev(\vec{p}^{(\ell)}, \vec{S}^{(\ell)}) = \gamma \alpha Rev(\vec{p}^{(\ell)}, \vec{S}^{(\ell)}).$$

The revenue bound follows.

\paragraph{Proof of Bicriteria Bound} Suppose that $\ell_2$ is the smallest index at which $Rev(\vec{p}^{(\ell_2)}, \vec{S}^{(\ell_2)}) \geq (1-\frac{1}{c})\frac{SW(\vec{S^{(1)}})}{\gamma \alpha}.$ Since $c \geq 1$, we know from the first part of this claim that there must always exist an $\ell_2$ satisfying the above inequality. Now, we can directly use the charging property to bound $SW(\vec{S}^{(\ell_2)})$. Let us add up the inequality  $SW(\vec{S}^{(t)}) - SW(\vec{S}^{(t+1)}) \leq \alpha Rev(\vec{p}^{(t)}, \vec{S}^{(t)}),$ from $t=1$ to $t=\ell_2 - 1$. We get that

\begin{align*}
SW(\vec{S^{(1)}}) - SW(\vec{S^{(\ell_2)}}) & \leq \alpha \sum_{t=1}^{\ell_2 - 1}Rev(\vec{p}^{(t)}, \vec{S}^{(t)}) \\
& \leq \alpha \sum_{t=1}^{\ell_2 - 1} (1-\frac{1}{c})\frac{SW(\vec{S^{(1)}})}{\gamma \alpha} \\
& \leq \gamma \alpha (1-\frac{1}{c})\frac{SW(\vec{S^{(1)}})}{\gamma \alpha}\\
& = (1-\frac{1}{c}) SW(\vec{S^{(1)}}).
\end{align*}

The second inequality comes from the definition of $\ell_2$ according to which $Rev(\vec{p}^{(t)}, \vec{S}^{(t)}) < (1-\frac{1}{c})\frac{SW(\vec{S^{(1)}})}{\gamma \alpha}$ for all $t \leq \ell_2 - 1$. This completes the proof.
\end{proof}
%

%
%

\subsubsection*{Proof of Claim~\ref{clm_frame2rev}}
%
%

\begin{proof}
Suppose that for some $r \geq 1$, $\vec{B}^{(r)}$ represents the allocation with the maximum social welfare among the sequence of benchmark allocations, i.e., $r = \argmax_{t=1}^\gamma SW(\vec{B}^{(t)})$. Define $\ell$ to be the index in the range $[1,\gamma]$ at which the pricing solution obtains the maximum revenue. So, for all $1 \leq t \leq \gamma$, $Rev(\vec{p}^{(\ell)}, \vec{S}^{(\ell)}) \geq Rev(\vec{p}^{(t)}, \vec{S}^{(t)})$. As per Definition~\ref{defn_framework2}, we know that for any $r \leq t \leq \gamma-1$, the following property holds,

$$SW(\vec{B}^{(t)}) - SW(\vec{B}^{(t+1)}) \leq \alpha Rev(\vec{p}^{(t)}, \vec{S}^{(t)}) + \frac{SW(\vec{B}^{(t)})}{\beta}.$$

Summing up the above sequence of inequalities from $t=r$ to $t=\gamma-1$ along with the inequality $SW(\vec{B}^{(\gamma)}) \leq \alpha Rev(\vec{p}^{(\gamma)}$, we get that

\begin{align*}
SW(\vec{B}^{(r)}) 
& \leq \alpha \sum_{t=1}^{\gamma} Rev(\vec{p}^{(t)}, \vec{S}^{(t)}) + \sum_{t=1}^{\gamma-1} \frac{SW(\vec{B}^{(t)})}{\beta}\\
& \leq \gamma\alpha  Rev(\vec{p}^{(\ell)}, \vec{S}^{(\ell)}) + (\gamma-1) \frac{SW(\vec{B}^{(r)})}{\beta}\\
\end{align*}

\noindent Therefore, after transposition, we have that $Rev(\vec{p}^{(\ell)}, \vec{S}^{(\ell)}) \geq \frac{1}{\gamma\alpha}SW(\vec{B}^{(r)})\{1-\frac{\gamma-1}{\beta}\} \geq \frac{1}{\gamma\alpha}SW(\vec{B}^{(1)})\{1-\frac{\gamma-1}{\beta}\}.$
\end{proof}
%
%

\subsubsection*{Proof of Claim~\ref{clm_frame2bicrit}}
%
%
\begin{proof}
We already know from Claim~\ref{clm_frame2rev} that there exists at least one pricing solution whose revenue is at least $\frac{SW(\vec{B}^{(1)})}{\gamma\alpha}(1-\frac{\gamma-1}{\beta})$. Without loss of generality, suppose that $\vec{B}^{(1)}$ is the benchmark allocation that maximizes social welfare. Let $\ell$ denote the smallest index $t$ such that $Rev(\vec{p}^{(t)},\vec{S}^{(t)}) \geq \frac{1}{2\gamma\alpha }SW(\vec{B}^{(1)})\{1-\frac{\gamma-1}{\beta}\}.$

Applying the telescoping summation argument as in the proof of Claim~\ref{clm_frame2rev} (from $t=1$ to $t=\ell -1$), and using the fact that $Rev(\vec{p}^{(t)},\vec{S}^{(t)}) < \frac{SW(\vec{B}^{(1)})}{2\gamma\alpha }(1-\frac{\gamma-1}{\beta})$ for all $t < \ell$, we get that

\begin{align*}
SW(\vec{B}^{(1)}) - SW(\vec{B}^{(\ell)}) & \leq \alpha \sum_{t=1}^{\ell-1}Rev(\vec{p}^{(t)},\vec{S}^{(t)}) +  \sum_{t=1}^{\ell-1} \frac{SW(\vec{B}^{(t)})}{\beta}\\
& \leq \frac{(\ell-1)}{2\gamma}SW(\vec{B}^{(1)})\{1-\frac{\gamma-1}{\beta}\} + \frac{\ell-1}{\beta}SW(\vec{B}^{(1)})\\
& \leq \frac{1}{2}SW(\vec{B}^{(1)})\{1-\frac{\gamma-1}{\beta}\} + \frac{\gamma-1}{\beta}SW(\vec{B}^{(1)}).
\end{align*}

Recall from the claim statement that $SW(\vec{B}^{(\ell)}) \leq  [\alpha Rev(\vec{p}^{(\ell)},\vec{S}^{(\ell)}) + Surp(\vec{p}^{(\ell)},\vec{S}^{(\ell)})](\frac{\beta}{\beta-1}).$ Adding this inequality to the above expression and moving all of the terms containing $SW(\vec{B}^{(1)})$ to the LHS, we get that

\begin{equation}
\frac{1}{2}SW(\vec{B}^{(1)})(1-\frac{\gamma-1}{\beta}) \leq [\alpha Rev(\vec{p}^{(\ell)},\vec{S}^{(\ell)}) + Surp(\vec{p}^{(\ell)},\vec{S}^{(\ell)})](\frac{\beta}{\beta-1}).
\label{eqn_f2bicriteriaintermed}
\end{equation}

By definition, $Rev(\vec{p}^{(\ell)},\vec{S}^{(\ell)}) \geq \frac{1}{2\alpha \gamma}SW(\vec{B}^{(1)})\{1-\frac{\gamma-1}{\beta}\}$. Suppose that $Rev(\vec{p}^{(\ell)},\vec{S}^{(\ell)}) = \frac{c}{2\alpha \gamma}SW(\vec{B}^{(1)})\{1-\frac{\gamma-1}{\beta}\}$ for some $c \geq 1$. Substituting this into Equation~\ref{eqn_f2bicriteriaintermed}, we get that

$$Surp(\vec{p}^{(\ell)},\vec{S}^{(\ell)}) \geq SW(\vec{B}^{(1)})(1-\frac{\gamma-1}{\beta})\{\frac{1}{2}(1-\frac{1}{\beta}) - \frac{c}{2\gamma}\}.$$

We know that $SW(\vec{S}^{(\ell)}) = Rev(\vec{p}^{(\ell)},\vec{S}^{(\ell)}) + Surp(\vec{p}^{(\ell)},\vec{S}^{(\ell)})$. Inserting the exact form for $Rev(\vec{p}^{(\ell)},\vec{S}^{(\ell)})$ in terms of $c$ and using the above lower bound for $Surp(\vec{p}^{(\ell)},\vec{S}^{(\ell)})$ gives us the claim.
\end{proof}

\section{Proof of Theorem \ref{thm.black-box}}
In order to keep the notation consistent, we define $(\vec{p}^{(0)}, \vec{S}^{(0)})$ to be the outcome of the algorithm on the original instance (or equivalently with reserve price $r=0$). Without loss of generality, we assume that $\gamma$ is an integer. We begin some easy propositions concerning the solutions returned by the algorithm in the face of reserve prices.

\begin{proposition}
\label{prop_pricelowerbound}
Given an instance $G$ and reserve price $r$, consider the outcome $(\vec{p}(r), \vec{S}(r))$ of $Alg$ for the instance $G(r)$ minus the dummy buyers. For every good $i \in \mathcal{I}$, $p_i(r) \geq r$.
\end{proposition}
\begin{proof}
Assume by contradiction that $p_i(r) < r$ for some $i \in \mathcal{I}$. We know that at this price, all $k_i+1$ dummy buyers corresponding to this good would receive strictly positive utility by purchasing the good but only $k_i$ units of the good are available to be consumed. Therefore, the solution returned by $Alg$ would not be a valid outcome of the simultaneous mechanism for instance $G(r)$, which is a contradiction.
\end{proof}

\begin{proposition}
\label{prop_saturatedsoldout}
Given an instance $G$ and reserve price $r$, consider the outcome $(\vec{p}(r), \vec{S}(r))$ of $Alg$ for the instance $G(r)$ minus the dummy buyers.
Suppose that for some good $i \in \mathcal{I}$, $p_i(r) > r$. Then, the good is sold out in $\vec{S}(r)$, i.e., $k_i(\vec{S}(r)) = k_i$.
\end{proposition}
\begin{proof}
First of all observe that since $p_i(r) > r$, the algorithm would not allocate this good to any of the dummy buyers. However, since the algorithm is \emph{locally welfare maximizing}, we know that if there is some good $i'$ that is not sold out, then the algorithm cannot improve its welfare by greedily allocating this good to some buyer with a null allocation. Therefore, the contrapositive of this property implies that if the algorithm has not allocated good $i$ to any of the dummy buyers (who have null allocations), then it is sold out.
\end{proof}

Suppose that $(\vec{p}(r), \vec{S}(r))$ represents the outcome of the augmented problem for some instance $G$ and reserve $r$. With respect to this outcome, we say that a good $i$ is \emph{saturated} if its price $p_i(r) > r$. We define $Sat(r) \subseteq \mathcal{I}$, to define the set of saturated goods in the pricing solution $(\vec{p}(r), \vec{S}(r))$.

\begin{proposition}
\label{prop_revsaturated}
Given an instance $G$ and price $r$, consider the outcome $(\vec{p}(r), \vec{S}(r))$ of $Alg$ for the instance $G(r)$ minus the dummy buyers. Then,
$Rev(\vec{p}(r), \vec{S}(r)) \geq \sum_{i \in Sat(r)}p_ik_i(\vec{S}({(r)})) = \sum_{i \in Sat(r)}p_ik_i$.
\end{proposition}

The proposition follows from the fact that all of the saturated goods are sold out. Our final proposition concerns the sequence of solutions generated by our algorithm, specifically the ratio of prices of goods at successive values of $t$.

\begin{proposition}
\label{prop_consecutivepricedoubling}
Consider the sequence of solutions $(\vec{p}^{(t)}, \vec{S}^{(t)})_{t=1}^{\gamma}$. For a given $1 \leq t \leq \gamma - 1$, suppose that some good $i \in \mathcal{I}$ is not saturated with respect to $(\vec{p}^{(t+1)}, \vec{S}^{(t+1)})$. Then, we have that $p^{(t+1)}_i \leq 2p^{(t)}_i$.
\end{proposition}
\begin{proof}
The proof is not hard to see. Since this good is unsaturated with respect to index $t+1$, we know that $p^{(t+1)}_i = r_{t+1}$. However, from Proposition~\ref{prop_pricelowerbound}, we also know that $p^{(t)}_i \geq r_t$. The claim follows since $r_{t+1}= 2r_t$.
\end{proof}

Now, we are ready to prove the charging property, namely $(i)$ for all $t \in [1,\gamma-1]$, $SW(S^{(t)}) - SW(\vec{S}^{(t+1)}) \leq  2 Rev(\vec{p}^{(t)}, \vec{S}^{(t)})$, and $(ii)$ $SW(S^{(\gamma)}) \leq 2 Rev(\vec{p}^{(\gamma)}, \vec{S}^{(\gamma)})$. Fix some value of $t < \gamma$ and consider some buyer $j$. At prices $\vec{p}^{(t+1)}$, this buyer receives her maximum utility from bundle $S^{(t+1)}_j$ as opposed to any other bundle including $S^{(t)}_j$. That is, we have that $v_j(S^{(t+1)}_j) - \sum_{i \in S^{(t+1)}_j}p^{(t+1)}_i \geq v_j(S^{(t)}_j) - \sum_{i \in S^{(t)}_j}p^{(t+1)}_i.$
Rearranging the inequality, and adding up the difference in welfare over all of the buyers, we get that

\begin{align*}
SW(\vec{S}^{(t)}) - SW(\vec{S}^{(t+1)}) 
& \leq \sum_{j \in \mathcal{N}}[\sum_{i \in S^{(t)}_j}p^{(t+1)}_i - \sum_{i \in S^{(t+1)}_j}p^{(t+1)}_i]\\
& = \sum_{i \in \mathcal{I}}p^{(t+1)}_i k_i(\vec{S}^{(t)}) - Rev(\vec{p}^{(t+1)}, \vec{S}^{(t+1)}) \\
& = \sum_{i \in \mathcal{I} \setminus Sat(r_{t+1})}p^{(t+1)}_i k_i(\vec{S}^{(t)}) + \sum_{i \in Sat(r_{t+1})}p^{(t+1)}_i k_i(\vec{S}^{(t)}) - Rev(\vec{p}^{(t+1)}, \vec{S}^{(t+1)}). \\
& \leq \sum_{i \in \mathcal{I} \setminus Sat(r_{t+1})}2p^{(t)}_i k_i(\vec{S}^{(t)}) + Rev(\vec{p}^{(t+1)}, \vec{S}^{(t+1)}) - Rev(\vec{p}^{(t+1)}, \vec{S}^{(t+1)})\\
& \leq 2Rev(\vec{p}^{(t)}, \vec{S}^{(t)}).
\end{align*}

Recall that $Sat(r_{t+1})$ refers to the set of saturated goods in the solution $(\vec{p}^{(t+1)}, \vec{S}^{(t+1)})$. Look at the penultimate step in the above set of inequalities. The upper bound for the goods in $\mathcal{I} \setminus Sat(r_{t+1})$ comes from Proposition~\ref{prop_consecutivepricedoubling}, which highlights the price doubling aspect of our reserve prices. The upper bound for the saturated goods is a bit more subtle. We know from Proposition~\ref{prop_saturatedsoldout} that these goods are fully sold out in $(\vec{p}^{(t+1)}, \vec{S}^{(t+1)})$. Using this in conjunction with Proposition~\ref{prop_revsaturated}, we get that $\sum_{i \in Sat(r_{t+1})}p^{(t+1)}_i k_i(\vec{S}^{(t)}) \leq \sum_{i \in Sat(r_{t+1})}p^{(t+1)}_i k_i \leq Rev(\vec{p}^{(t+1)}, \vec{S}^{(t+1)}).$

We have proved the first of the charging properties. Consider the $t=\gamma$ case, where the reserve price $r_{\gamma} = 2^{\log (mk\alpha)} \frac{SW(\vec{S}^{(0)})}{2 mk} \geq \frac{SW(OPT)}{2}$ since $SW(\vec{S}^{(0)})$ is an $\alpha$-approximation to $SW(OPT)$ by definition of $Alg$. For this case, we need to prove that $SW(\vec{S}^{(\gamma)}) \leq 2 Rev(\vec{p}^{(\gamma)}, \vec{S}^{(\gamma)})$. Without loss of generality, suppose that $\vec{S}^{(\gamma)}$ is not a null allocation. Since at least one good (say good $i$) is purchased by the buyers, we have that $Rev(\vec{p}^{(\gamma)}, \vec{S}^{(\gamma)}) \geq p^{(\gamma)}_i \geq r_{\gamma} \geq \frac{SW(OPT)}{2} \geq \frac{SW(\vec{S}^{(\gamma)})}{2}$. The charging property follows.

Now, we can directly apply Claim~\ref{clm_simplecharging} with $c=\frac{3}{2}$ to get the following lower bounds on the social welfare of the solution returned by our black-box algorithm,

$$Rev(\vec{p}^{(\ell)}, \vec{S}^{(\ell)}) \geq \frac{SW(\vec{S}^{(1)})}{6\gamma} \quad \text{and} \quad SW(\vec{S}^{(\ell)}) \geq \frac{2SW(\vec{S}^{(1)})}{3}.$$

In order to complete the theorem, we only need to show a lower bound on $SW(\vec{S}^{(1)})$ in terms of the optimum social welfare, which we show in the next lemma.

\begin{lemma}
$$SW(OPT) \leq 2\alpha SW(\vec{S}^{(1)}).$$
\end{lemma}
\begin{proof}
Recall that $(\vec{p}^{(0)}, \vec{S}^{(0)})$ is the solution returned by $Alg$ on the original instance. Since the algorithm returns an $\alpha$-approximation to OPT, we also know that $SW(\vec{S}^{(0)}) \geq \frac{SW(OPT)}{\alpha}$. Now, consider the prices $\vec{p}^{(1)}$. At these prices, each buyer $j$ receives more utility from the bundle $S^{(1)}_j$ than she would from any other bundle including $S^{(0)}_j$. Also, recall that the reserve price corresponding to the solution $(\vec{p}^{(1)}, \vec{S}^{(1)})$ is $r_1 = \frac{SW(\vec{S}^{(0)})}{2 mk}$. Using the same upper bound for the difference in welfare in terms of prices as we did before, we get that

\begin{align*}
SW(\vec{S}^{(0)}) - SW(\vec{S}^{(1)}) & \leq \sum_{j \in \mathcal{N}}[\sum_{i \in S^{(0)}_j}p^{(1)}_i - \sum_{i \in S^{(1)}_j}p^{(1)}_i]\\
& = \sum_{i \in \mathcal{I} \setminus Sat(r_{1})}p^{(1)}_i k_i(\vec{S}^{(0)}) + \sum_{i \in Sat(r_{1})}p^{(1)}_i k_i(\vec{S}^{(0)}) - Rev(\vec{p}^{(1)}, \vec{S}^{(1)}). \\
& \leq \sum_{i \in \mathcal{I} \setminus Sat(r_{1})}r_1 k_i(\vec{S}^{(0)}) + Rev(\vec{p}^{(1)}, \vec{S}^{(1)}) - Rev(\vec{p}^{(1)}, \vec{S}^{(1)})\\
& \leq \frac{SW(\vec{S}^{(0)})}{2 mk} mk = \frac{SW(\vec{S}^{(0)})}{2}.
\end{align*}

Since $SW(\vec{S}^{(0)}) - SW(\vec{S}^{(1)}) \leq \frac{SW(\vec{S}^{(0)})}{2}$, we get that $SW(\vec{S}^{(1)}) \geq \frac{SW(\vec{S}^{(0)})}{2} \geq \frac{SW(OPT)}{2\alpha}$. Note that in the above series of inequalities, we used the fact that for all of the unsaturated goods, $p^{(1)}_i = r_1 = \frac{SW(\vec{S}^{(0)})}{2 mk}$. The upper bound for the saturated goods comes from Proposition~\ref{prop_revsaturated}.
\end{proof}

The theorem follows.

\section{Proof of Theorem~\ref{thm:implications}}

\begin{proof}
\textbf{General Pricing Mechanism:}  
Since the results in the final statement are not based on envy-free item pricing, we carefully redefine our simple charging framework to include a broader class of mechanisms although it is not hard to see from Definition~\ref{defn_framework2} and Claims~\ref{clm_frame2rev},~\ref{clm_simplecharging} that the property and its implications are not really tied to any specific type of mechanism. Specifically, consider a general mechanism $\mathcal{M}$ whose output $(\vec{P},\vec{S})$ is an allocation of goods and prices charged to each buyer, i.e., $P_j$ denotes the price that the mechanism charges to buyer $j \in \mathcal{N}$. For such mechanisms, we can still define $Rev(\vec{P}, \vec{S})$ analogously.

We now provide a very brief extension of the charging properties to more general mechanisms. However, since the exact same proofs carry over to this setting, we do not reprove entire claims. First of all, consider a sequence of solutions $(\vec{S}^{(t)}, \vec{P}^{(t)})_{t=1}^{\gamma}$, which could be the outcomes of several different mechanisms or the same mechanism for different instances. We say that this sequence satisfies the simple charging property if $(i)$ for all $t \leq \gamma - 1$, $SW(\vec{S}^{(t)}) - SW(\vec{S}^{(t+1)}) \leq \alpha Rev(\vec{P}^{(t)}, \vec{S}^{(t)}),$ and $(ii)$ $SW(\vec{S}^{(\gamma)}) \leq \alpha Rev(\vec{P}^{(\gamma)}, \vec{S}^{(\gamma)})$ for some $\alpha \geq 1$. It is then not hard to see that the following still hold simply by reproving Claims~\ref{clm_frame2rev} and~\ref{clm_simplecharging} with the new notation.
\begin{enumerate}
\item If a sequence of solutions satisfies the charging property, then there exists $\ell_1 \in [1,\gamma]$ such that $Rev(\vec{P}^{(\ell_1)}, \vec{S}^{(\ell_1)}) \geq \frac{SW(\vec{S^{(1)}})}{\gamma \alpha}$

\item If a sequence of solutions satisfies the charging property, then there exists $\ell_2 \in [1,\gamma]$ such that $Rev(\vec{P}^{(\ell_2)}, \vec{S}^{(\ell_2)}) \geq \frac{SW(\vec{S^{(1)}})}{3\gamma \alpha}$ and at the same time $SW(\vec{S}^{(\ell_2)}) \geq \frac{2SW(\vec{S^{(1)}})}{3}.$
\end{enumerate}

We now prove the actual theorem in three cases.

\begin{enumerate}
\item For settings with gross substitutes valuations, it is well known~\cite{gul1999walrasian} that there exists a optimal item pricing based simultaneous mechanism for welfare maximization, i.e., $\alpha=1$. Plugging this into Theorem~\ref{thm.black-box} directly yields the desired bound.

\item The proof of the central result in~\cite{cheungS08} proceeds by computing a sequence of $\log k_{max}$ pricing solutions that satisfy the simple charging property for a constant $\alpha$. The statement follows from a direct application of Claim~\ref{clm_simplecharging} for $c=2$.

\item The final statement is a bit tricky so we carefully prove both parts. First, consider the bundle pricing mechanism. The algorithm from~\cite{feldmanGL16} computes a sequence of $\gamma= 2+\log(N)$ pricing solutions $(\vec{S}^{(t)}, \vec{P}^{(t)})_{t=1}^{\gamma}$ satisfying $(i)$ $\frac{SW(\vec{S}^{(1)})}{2} - SW(\vec{S}^{(2)}) \leq Rev(\vec{P}^{(1)}, \vec{S}^{(1)}),$, $(ii)$ $SW(\vec{S}^{(t)}) - SW(\vec{S}^{(t+1)}) \leq 2Rev(\vec{P}^{(t)}, \vec{S}^{(t)})$ for $t=1$ to $t=\gamma-1$, and $(iii)$ $SW(\vec{S}^{(\gamma)}) \leq 2Rev(\vec{P}^{(\gamma)}, \vec{S}^{(\gamma)})$. Finally, suppose that we are given an $\alpha^*$-approximation algorithm for the given class of valuations, it is known from~\cite{feldmanGL16} that $SW(\vec{S}^{(1)}) \geq \frac{SW(OPT)}{2\alpha^*}$. The statement follows from a slight variant of Claim~\ref{clm_simplecharging}.

Now, consider the non envy-free item pricing mechanism. Once again, from~\cite{elbassioniFS10}, we get a sequence of $\gamma=k_{max}$ (fractional) pricing solutions $(\vec{S}^{(t)}, \vec{P}^{(t)})_{t=1}^{\gamma}$ such that $(i)$ for all $1 \leq t \leq \gamma -1$, we have that $SW(\vec{S}^{(t)}) - SW(\vec{S}^{(t+1)}) \leq 4(1+\epsilon) Rev(\vec{P}^{(t+1)}, \vec{S}^{(t+1)})$ for some suitable choice of $\epsilon$. Secondly $SW(\vec{S}^{(\gamma)}) \leq 2Rev(\vec{P}^{(\gamma)}, \vec{S}^{(\gamma)})$. Suppose that $\ell$ is the smallest index where $Rev(\vec{P}^{(\ell)}, \vec{S}^{(\ell)}) \geq \frac{SW(\vec{S}^{(1)})}{12\gamma(1+\epsilon)}$. Using the same ideas as Claim~\ref{clm_frame2bicrit}, we perform a telescoping summation and get that,

$$SW(\vec{S}^{(1)}) - SW(\vec{S}^{(\ell)}) \leq \frac{SW(\vec{S}^{(1)})}{3} + Rev(\vec{P}^{(\ell)}, \vec{S}^{(\ell)}).$$

Since, $Rev(\vec{P}^{(\ell)}, \vec{S}^{(\ell)}) \leq SW(\vec{S}^{(\ell)})$, we finally get that $SW(\vec{S}^{(\ell)}) \geq \frac{SW(\vec{S}^{(1)})}{3}$, where $SW(\vec{S}^{(1)})$ is the social welfare of the optimum fractional allocation. Finally, the authors in the paper provide a technique to convert any fractional solution $(\vec{S}^{(t)}, \vec{P}^{(t)})$ to an integral one with an additional factor $2$ loss in both profit and welfare. Therefore, the integral solution obtained upon rounding  $\vec{S}^{(\ell)})$ has a social welfare of at least $\frac{SW(\vec{S}^{(1)})}{6}$.

\end{enumerate}
%
%

\end{proof}

\section{Proofs from Section~\ref{sec:itemhalving}}
\label{app:sec4}

\subsubsection*{Proof of Claim~\ref{clm_case1gen}}
\begin{proof}
By definition, the social welfare of the allocation $\vec{B}$ is $\sum_{i \in \mathcal{I}}U_i(\vec{B})$. We will show lower bounds on both the revenue and buyer surplus achieved by the sequential posted pricing mechanism defined in the statement of this claim. Suppose that $Sold$ denotes the set of goods that are sold out during the course of the mechanism, i.e., $k_i(\vec{S}) = q_i$. The total revenue obtained by this mechanism is at least the revenue due to the goods in $Sold$, and so we have that

\begin{equation}
Rev(\vec{p},\vec{S}) \geq \sum_{i \in Sold}p_iq_i = \sum_{i \in Sold}\frac{U_i(\vec{B})}{\alpha q_i} q_i = \sum_{i \in Sold}\frac{U_i(\vec{B})}{\alpha}.
\label{eqn_revlowerbound}
\end{equation}

Moving on to buyer utility, consider any buyer $j$. We know that this buyer purchases a utility-maximizing bundle from the set of goods available in the round where she arrives. Clearly, this buyer's surplus (call it $Surp_j(\vec{p},\vec{S})$) cannot be smaller than the utility that she could have derived by purchasing any other subset of goods offered to her by the mechanism. Specifically, we know that at least one copy of each of the goods in $\mathcal{I} \setminus Sold$ is available during every round of the mechanism including that where buyer $j$ arrives. Therefore, this buyer's surplus is at least her utility for the goods in $B'_j:=(\mathcal{I} \setminus Sold) \cap B_j$. So, we have that $Surp_j(\vec{p},\vec{S}) \geq v_j(B'_j) - \sum_{i \in B'_j}p_i \geq \sum_{i \in B'_j}[x^{B_j}_j(i) - p_i]$. Here we used the XoS property that $v_j(B'_j) \geq x^{B_j}_j(B'_j) = \sum_{i \in B'_j}x^{B_j}_j(i)$. Summing this quantity over all buyers, we get that

\begin{align*}
Surp (\vec{p},\vec{S}) & \geq \sum_{j \in \mathcal{N}} \sum_{i \in (\mathcal{I} \setminus Sold) \cap B_j}[x_j^{B_j}(i) - p_i]\\
& = \sum_{i \notin Sold}[\sum_{j \in N_i(\vec{B})}x^{B_j}_j(i) - p_i q_i]\\
& = \sum_{i \notin Sold}U_i(\vec{B}) - \frac{U_i(\vec{B})}{\alpha}\\
& = \sum_{i \notin Sold}U_i(\vec{B})(1-\frac{1}{\alpha}).
\end{align*}

Combining our lower bounds for revenue and surplus, we are now ready to complete our proof.
\begin{align*}
SW(\vec{B}) - \frac{SW(\vec{B})}{\alpha} & = \sum_{i \in \mathcal{I}}U_i(\vec{B})(1-\frac{1}{\alpha})\\
& = \sum_{i \in Sold}U_i(\vec{B})(1-\frac{1}{\alpha}) + \sum_{i \notin Sold}U_i(\vec{B})(1-\frac{1}{\alpha}).\\
& \leq \alpha Rev(\vec{p},\vec{S})  + Surp(\vec{p},\vec{S}).
\end{align*}

After some mild re-arranging, we get the claim.
\end{proof}

\subsubsection*{Proof of Claim~\ref{clm_allocunsoldgen}}

\begin{proof}
For each good $i \in \mathcal{I}$, define $r_i := k_i(\vec{S}^M).$ By definition, $r_i \leq q_i$ since the mechanism cannot sell more than $q_i$ copies of good $i$. The revenue provided by the mechanism is given by

$$ Rev(\vec{p},\vec{S}^M) = \sum_{i \in \mathcal{I}}p_i k_i(\vec{S}^M) = \sum_{i \in \mathcal{I}}\frac{U_i(\vec{B})}{\alpha q_i} r_i.$$

We know that in the allocation $\vec{S}$, each good $i$ is allocated to the buyers in $Top_i(q_i - r_i, \vec{B})$, where we take $Top_i(0,\vec{B})$ to be the null set. Therefore, for each buyer $j$, we have that $S_j \subseteq B_j$. We now derive a lower bound on the social welfare of the allocation $\vec{S}$ using the XoS property.

\begin{align*}
SW(\vec{S}) & = \sum_{j \in \mathcal{N}}v_j(S_j) \\
& \geq \sum_{j \in \mathcal{N}}x^{B_j}_j(S_j) \\
& = \sum_{j \in \mathcal{N}} \sum_{i \in S_j} x^{B_j}_j(i)\\
& = \sum_{i \in \mathcal{I}}\sum_{j \in Top_i(q_i - r_i, \vec{B})} x^{B_j}_j(i).\\
\end{align*}

Combining our lower bound for the social welfare of $\vec{S}$ with the exact expression for the mechanism's revenue, we get that
\begin{align*}
\alpha Rev(\vec{p},\vec{S}^M) + SW(\vec{S}) & \geq \sum_{i \in \mathcal{I}}\frac{U_i(\vec{B})}{q_i} r_i + \sum_{i \in \mathcal{I}}\sum_{j \in Top_i(q_i - r_i, \vec{B})} x^{B_j}_j(i)\\
& = \sum_{i \in \mathcal{I}}\{\frac{U_i(\vec{B})}{q_i} r_i + \sum_{j \in Top_i(q_i - r_i, \vec{B})} x^{B_j}_j(i)\}.
\end{align*}

We now resort to Lemma~\ref{lem_propertyaverage} to bound the RHS in terms of $U_i(\vec{B})$. Specifically, we know that $U_i(\vec{B}) = \sum_{j \in N_i(\vec{B})}x^{B_j}_j(i)$ and that $Top_i(q_i-r_i,\vec{B}) \subseteq N_i(\vec{B})$ represents the $q_i - r_i$ buyers who contribute the most to $U_i(\vec{B})$. Therefore, applying Lemma~\ref{lem_propertyaverage} with $n=k_i(\vec{B})$\footnote{We simply ignore goods for which $k_i(\vec{B}) = 0$ since they do not contribute to the welfare of any of the solutions that we consider here} and $r = r_i$, $V=U_i(\vec{B})$, we get that for every $i \in \mathcal{I}$,

$$\frac{U_i(\vec{B})}{q_i} r_i + \sum_{j \in Top_i(q_i - r_i, \vec{B})} x^{B_j}_j(i) \geq U_i(\vec{B}).$$

In conclusion, we have that $\alpha Rev(\vec{p},\vec{S}^M) + SW(\vec{S}) \geq \sum_{i \in \mathcal{I}} U_i(\vec{B}) = SW(\vec{B})$ and the claim follows. \end{proof}

\subsubsection*{Proof of Theorem~\ref{thm_bicriteriaxos}}
\begin{proof}
The proof can be divided into two high-level segments. Consider the sequence of solutions returned by our core algorithm for the above mentioned inputs: $\vec{B}^ {(1)}, \ldots, \vec{B}^ {(\gamma)}$ and $(\vec{p}^ {(1)}, \vec{S}^{(1)}) \ldots, (\vec{p}^ {(\gamma)},\vec{S}^{(\gamma)})$. We have already proved that these solutions satisfy the conditions described in Definition~\ref{defn_framework2}. Recall that $\vec{S}^{(t)}$ is the allocation obtained upon running the sequential mechanism with the parameters $(\vec{p}^{(t)}, k(\vec{B}^{(t)}))$ and order of arrival $\Gamma(\vec{p}^{(t)}, \vec{B}^{(t)})$, and let $Rev^{(t)}$, $Surp^{(t)}$ refer to the revenue and surplus achieved by this mechanism, i.e., $Rev^{(t)}=Rev(\vec{p}^{(t)},\vec{S}^{(t)})$ and $Surp^{(t)}=Surp(\vec{p}^{(t)},\vec{S}^{(t)})$. In the first half of the proof, we will show that in addition to the generalized charging property, our solutions also satisfy the property mentioned in Claim~\ref{clm_frame2bicrit}, i.e., for all $1 \leq t \leq \gamma$ and $\alpha = \beta = 2\gamma$,

$$SW(\vec{B}^{(t)})(1-\frac{1}{\beta}) \leq \alpha Rev^{(t)} + Surp^{(t)}.$$

One can then apply Claim~\ref{clm_frame2bicrit} to obtain a continuous bicriteria approximation parameterized on an instance-specific measure $c \geq 1$. In the second part, we will select appropriate threshold values for this parameter $c$ and partition our bicriteria approximation space into disjoint regions to yield the guarantees made in the claim. We begin by showing that our solutions satisfy the new property mentioned in Claim~\ref{clm_frame2bicrit}.

\begin{lemma}
\label{lem_spec_bicrit_property}
For all $1 \leq t \leq \gamma$, the following inequality holds:

$$SW(\vec{B}^{(t)})(1-\frac{1}{2\gamma}) \leq 2\gamma Rev^{(t)} + Surp^{(t)}.$$

\end{lemma}
\begin{proof}
The proof mostly follows from Claim~\ref{clm_case1gen} of Theorem~\ref{thm.XoScharging}. First suppose that $t < \gamma$, and apply the claim for $\vec{B} = \vec{B}^{(t)}$, $\alpha = 2\gamma$. We know that the output of the mechanism as per the claim corresponds to the allocation $\vec{S}^{(t)}$ here and its revenue, surplus correspond to $Rev^{(t)}, Surp^{(t)}$ respectively. Therefore, from the claim, we get that

$$SW(\vec{B}^{(t)}) - Surp^{(t)} \leq 2\gamma Rev^{(t)} + \frac{SW(\vec{B}^{(t)}) }{2\gamma}.$$

A little rearrangement gives us the inequality claimed in the lemma. Next, suppose that $t = \gamma$. In this case, due to Lemma \ref{lem_swgammaupper}, and the revenue obtained by $(\vec{p}^{(\gamma)},\vec{S}^{(\gamma)})$ is $v_{j^*}(i^*)-\epsilon$ (see Figure \ref{fig:tail}), we know that for a small enough epsilon it must be that $SW(\vec{B}^{(\gamma)}) \leq 2\gamma Rev^{(\gamma)}$, as desired.
\end{proof}

Now that we have proved the lemma, we can directly apply Claim~\ref{clm_frame2bicrit} substituting $\gamma = \log m + \log k$ and $\alpha = \beta = 2\gamma$ to obtain the following continuous bicriteria guarantees dependent on an instance-specific parameter $c \geq 1$. Here $\ell$ is the index defined in the algorithm.

\begin{align}
\label{eqn_swbic} (i) SW(\vec{S}^{(\ell)}) & \geq \frac{SW(\vec{A})}{4}\{(1-\frac{1}{2\gamma}) - \frac{c}{\gamma} + \frac{c}{2\gamma^2}\}\\
\label{eqn_revbic} (ii)  Rev(\vec{p}^{(\ell)},\vec{S}^{(\ell)}) & \geq \frac{c}{8\gamma^2}SW(\vec{A}).
\end{align}

To avoid messy notation, we have expressed some of the terms as a function of $\gamma$. By definition $c \geq 1$. Without loss of generality, we also take $\gamma$ to be greater than $2$. The rest of the proof is presented in two cases depending on the ratio of $c$ and $\gamma$ in the given instance,

\begin{enumerate}

\item When \textit{$c \leq \frac{\gamma}{2} - \frac{1}{2}$.} In this case, substituting $c \geq 1$, gives us the following trivial bound for revenue, $Rev(\vec{p}^{(\ell)},\vec{S}^{(\ell)}) \geq \frac{1}{8(\log m + \log k)^2}SW(\vec{A}).$ Next, since $c \leq \frac{\gamma}{2} - \frac{1}{2}$, we also have that $1-\frac{1}{2\gamma} - \frac{c}{\gamma} + \frac{c}{2\gamma^2} \geq \frac{1}{2}$, and so $SW(\vec{S}^{(\ell)}) \geq \frac{SW(\vec{A})}{8}$.

\item When \textit{$c \geq \frac{\gamma}{2} - \frac{1}{2}$.}  Here we directly use the expression for $Rev(\vec{p}^{(\ell)},\vec{S}^{(\ell)})$ and the fact that social welfare cannot be smaller than revenue to get

$$ Rev(\vec{p}^{(\ell)},\vec{S}^{(\ell)})  \geq (\frac{1}{2} - \frac{1}{2\gamma})\frac{1}{8\gamma}SW(\vec{A}).$$

Therefore, $SW(\vec{A}) \leq \Theta(\log m + \log k)Rev(\vec{p}^{(\ell)},\vec{S}^{(\ell)}) \leq \Theta(\log m + \log k)SW(\vec{S}^{(\ell)}).$
\end{enumerate}

\noindent \textbf{Other Arrival Orders} Observe that $SW(\vec{S}^{(\ell)})$ is the social welfare of the outcome returned by the sequential posted pricing mechanism on a specific arrival order, i.e., the arrival order $\Gamma(\vec{p}^{(\ell)}, \vec{B}^{(\ell)})$ that results in the worst revenue in $\Pi$. Does the bicriteria dichotomy also hold for other arrival orders in $\Pi$? In order to show this, first fix some $\pi \in \Pi$. Suppose that $Rev^{(\ell)}(\pi), Surp^{(\ell)}(\pi)$, and $SW^{(\ell)}(\pi)$ denote the revenue, surplus, and social welfare respectively for the sequential posted pricing mechanism with parameters $(\vec{p}^{(\ell)} , \vec{q}^{(\ell)}=k(\vec{B}^{(\ell)}))$, when the buyer arrival order is given by $\pi$. Since Claim~\ref{clm_case1gen} does not depend on a specific arrival order, we know that as long as $\ell < \gamma$, $SW(\vec{B}^{(\ell)})(1-\frac{1}{2\gamma}) \leq 2\gamma Rev^{(\ell)}(\pi) + Surp^{(\ell)}(\pi).$ In the case that $\ell = \gamma$, it is not hard to deduce from the proof of Theorem~\ref{thm.XoScharging} that $SW(\vec{B}^{(\gamma)}) \leq 2\gamma Rev^{(\gamma)}(\pi)$. Therefore, when $\ell = \gamma$, the following inequality is trivially true: $SW(\vec{B}^{(\ell)})(1-\frac{1}{2\gamma}) \leq 2\gamma Rev^{(\ell)}(\pi) + Surp^{(\ell)}(\pi).$

We are now in a position to sketch the rest of the proof, which mostly follows from Claim~\ref{clm_frame2bicrit}. Looking at the proof of that claim, with $\beta = 2\gamma$, we get that $\frac{1}{4}SW(\vec{B}^{(1)}) \leq SW(\vec{B}^{(\ell)})$. Using the newly obtained upper bound for $SW(\vec{B}^{(\ell)})$ in terms of the revenue and surplus of the mechanism with arrival order $\pi$, we get that

$$\frac{1}{4}SW(\vec{B}^{(1)})(1-\frac{1}{2\gamma}) \leq 2\gamma Rev^{(\ell)}(\pi) + Surp^{(\ell)}(\pi).$$

Proceeding exactly as we did in the proof of Claim~\ref{clm_frame2bicrit}, we can conclude that Equations~\ref{eqn_swbic} and~\ref{eqn_revbic} are still applicable even when their left hand sides are replaced by $SW^{(\ell)}(\pi)$ and $Rev^{(\ell)}(\pi)$ respectively. The bicriteria trade-off for the arbitrary arrival order $\pi$ follows from the same dichotomous analysis based on the value of $c$, as in the current theorem.
\end{proof}

\section{Proofs from Section~\ref{sec:multi-item}}
\label{app:sec5}

\subsection*{Proof of Theorem \ref{thm:multi-unit}}

The algorithm that achieves this approximation guarantee is based on the same item halving core algorithm as the XoS case. The input to the core algorithm still consists of an allocation $\vec{A}$, and a parameter $\gamma \geq 1$ (in our case, $\gamma = \log m$). However, before we get to the main algorithm, there is a simple case which gives us an immediate revenue guarantee. Set $p' = \frac{SW(\vec{A})}{2\gamma \theta(\vec{A})}$ (this is exactly the function $Prices(\vec{A}, \gamma)$ as defined before). Let $\vec{S'}$ be the output allocation of the simultaneous mechanism with price $p'$ on all items. If $S'$ is not a valid solution (i.e., if too many buyers want the items, and thus result in taking more than $m$ items), then we argue below that we can construct a solution with high revenue easily.

\subsubsection*{Case I: If $\theta(\vec{S'}) > m$}
Define $p^{(0)} > p'$ to be the smallest price at which less than or equal to $m$ copies are purchased by buyers. Define $\vec{S}^{(0)}$ to be the output allocation of the simultaneous mechanism with price $p^{(0)}$ on all items; thus $(p^{(0)}, \vec{S}^{(0)})$ is a valid solution. Moreover, we want to make sure that $\vec{S}^{(0)}$ is locally maximal in the following sense: if a buyer $j$ is indifferent between a bundle of size $|S_j^{(0)}|$ which $j$ receives in $\vec{S}^{(0)}$ and a bundle of larger size $q>|S_j^{(0)}|$ at prices $p^{(0)}$, then it must be that the allocation resulting from altering $\vec{S}^{(0)}$ by giving $j$ a set of $q$ items instead of $|S_j^{(0)}|$ items would allocate strictly more than $m$ items, and thus would not be a valid solution (i.e., $\theta(\vec{S}^{(0)}) - |S_j^{(0)}| + q > m$). It is easy to see that such an allocation $\vec{S}^{(0)}$ can be found efficiently, by simply finding an arbitrary $\vec{S}^{(0)}$, and then considering each buyer $j$ and giving them a bigger bundle if this still results in a valid solution to the prices $p^{(0)}$.

In the case that $\theta(\vec{S'}) > m$, we argue here that $(p^{(0)}, \vec{S}^{(0)})$ has high revenue. Before proving the revenue bound for this case, we will state a simple claim about how buyers behave as the price increases. For a given non-negative integer $q$, we abuse notation and take $v_j(q)$ to be a buyer's valuation for purchasing $q$ units of the good.

\begin{proposition}
\label{prop_buyerdropbundle}
Suppose that there exists a multi-unit buyer $j$ and a price $p$ such that the buyer maximizes her utility at this price by purchasing $q_2$ units of the good, but for every sufficiently small $\epsilon> 0$, the buyer maximizes her utility at price $p - \epsilon$ by purchasing $q_1 \neq q_2$ units of the good. Then,
\begin{enumerate}
\item $q_1 > q_2$,

\item $v_j(q_1) - p|q_1| = v_j(q_2) - p|q_2|$.
\end{enumerate}
\end{proposition}
\begin{proof}
The first part is trivial since as we increase prices, a buyer's consumption cannot increase strictly. For the second part, assume by contradiction that $v_j(q_1) - p|q_1| < v_j(q_2) - p|q_2|$ (the inequality has to hold in this direction since $q_2$ is this buyer's utility maximizing consumption). Then, it is not hard to construct a sufficiently small $\epsilon > 0$, such that $v_j(q_1) - (p-\epsilon)|q_1| < v_j(q_2) - (p-\epsilon)|q_2|$, which contradicts the fact that at price $p - \epsilon$, the buyer maximizes her utility by purchasing quantity $q_1$.
\end{proof}

In other words, the proposition establishes that as we increase the price, each buyer encounters a point of indifference between a bundle of a larger size and one of smaller size. The following establishes our main result for this case.

\begin{claim}
The solution $(p^{(0)}, \vec{S}^{(0)})$, which is a valid outcome of the simultaneous mechanism, obtains a revenue that is within a logarithmic factor of the welfare of $\vec{A}$. Specifically,

$$Rev(p^{(0)}, \vec{S}^{(0)}) \geq \frac{SW(\vec{A})}{4\gamma}.$$
\end{claim}
\begin{proof}
By definition, no more than $m$ copies are allocated in this allocation so the outcome is valid. Second, we know that by definition $p^{(0)} \geq p' = \frac{SW(\vec{A})}{2\gamma\theta(\vec{A})} \geq \frac{SW(\vec{A})}{2\gamma m}$. Finally, we will use the no overwhelming buyer assumption and prove that at least $\frac{m}{2}$ copies are allocated in $\vec{S}^{(0)}$. To see why, observe that by definition of $p^{(0)}$, for sufficiently small $\epsilon > 0$, the simultaneous mechanism with price $p^{(0)} - \epsilon$ will allocate more than $m$ units of good, whereas at $p^{(0)}$, the mechanism allocates less than or equal to $m$ units. How small can $\theta(\vec{S}^{(0)})$ be?

Let $X$ be the set of buyers whose consumption dropped at price $p^{(0)}$ in order to meet the supply constraints, i.e., the amount consumed by these buyers at price $p^{(0)} - \epsilon$ is strictly larger than their consumption at price $p^{(0)}$. Applying Proposition~\ref{prop_buyerdropbundle}, we have that each buyer $j \in X$ is indifferent between her earlier consumption (call it $\bar{S}_j$) and consuming $|S^{(0)}_j|$ units at price $p^{(0)}$.

Recall our definition of $\vec{S}^{(0)}$ from earlier: we defined $\vec{S}^{(0)}$ to be locally maximal, meaning that for every buyer $j \in X$, we have that $\theta(\vec{S}^{(0)}) + |\bar{S}_j| - |S^{(0)}_j| > m$. If this were not the case, then we could have allocated this buyer $|\bar{S}_j|$ copies without violating the supply constraints, which contradicts our definition of $\vec{S}^{(0)}$. Therefore, we have that $\theta(\vec{S}^{(0)})  > m + |S^{(0)}_j| - |\bar{S}_j| \geq m - |\bar{S}_j| \geq \frac{m}{2}$. The last inequality comes from the no overwhelming buyer assumption, according to which $|\bar{S}_j| \leq \frac{m}{2}$. To bound the revenue, consider the following,

$$Rev(p^{(0)}, \vec{S}^{(0)}) = p^{(0)} \theta(\vec{S}^{(0)}) \geq  \frac{SW(\vec{A})}{2 \gamma m} \frac{m}{2} \geq \frac{SW(\vec{A})}{4\gamma}.$$
\end{proof}

\subsubsection*{Case II: If $\theta(\vec{S'}) \leq m$}
This is the more difficult case, in which we take advantage of the Item Halving ideas. The core algorithm is the same as before with a few minor changes. We only list the modifications and addenda to the algorithm.
\begin{description}

\item [Function Prices$(\vec{S}, \gamma)$] At each stage, the algorithm simply posts a single price on all copies of the good. Therefore, for the given inputs, the prices function can be interpreted as returning the price as $p =  \frac{SW(\vec{S})}{2\gamma \theta(\vec{S})}$. Suppose that $p^{(t)}$ is the price during iteration $t$.

\item [Simultaneous Mechanism] We no longer require the arrival function $\Gamma$ or a call to the sequential mechanism. Instead $\vec{S}^{(t)}$ is now the allocation returned by the simultaneous mechanism with price $p^{(t)}$ on all of the goods. It is possible that $\vec{S}^{(t)}$ may not be a valid outcome as more than $m$ copies of the good may be allocated: the mechanism simply gives each buyer their most preferred bundle independently. We show that the ultimate solution returned by our algorithm never violates the supply constraints.

\item [Function Alloc-Unsold] The allocation is constructed by providing each buyer $j$, $\max(|B_j| - |S_j|, 0)$ copies of goods. We will later see that the choice of the function does not really matter and in fact, the algorithm obtains a good revenue before the first call to this function. Moreover, we slightly alter the if-else definitions in the core algorithm, so that Alloc-Unsold is called when $\theta(\vec{B}^{(t)}) \leq 2\theta(\vec{S}^{(t)})$.
\end{description}

We begin by restating some notation. Recall that for the multi-unit case, the outcome of the simultaneous mechanism is simply a single price along with an allocation. In that regard, our algorithm computes a sequence of $\gamma$ allocations $(p^{(t)}, \vec{S}^{(t)})_{t=1}^{\gamma}$. We do remark that some of these allocations may actually not be \emph{valid outcomes} since more than $m$ copies may be allocated to buyers. However, we will show that for each instance, there exists at least one valid outcome with the desired revenue guarantee. 


Suppose that $\ell$ is the first time step at which our core algorithm invokes the AllocUnsold function, i.e., $\ell$ is the smallest index at which $\theta(\vec{B}^{(t)}) \leq 2\theta(\vec{S}^{(t)})$. In the boundary event that the function is never invoked, we define $\ell = \gamma$. The proof for this case will proceed as follows: first, we show that the sequence of solutions $(p^{(t)}, \vec{S}^{(t)})_{t=1}^{\ell}$ are all valid outcomes of the simultaneous mechanism\footnote{Note that by definition, each buyer is purchasing her utility maximizing set of goods, so that part of the definition of the simultaneous mechanism is never violated.}. That is, at each $1 \leq t \leq \ell$, $\theta(\vec{S}^{(t)}) \leq m$. Secondly, we will prove that $SW(\vec{B}^{(\ell)}) \geq \frac{SW(\vec{A})}{2}$. Finally, since at least half the number of goods allocated in $\vec{B}^{(\ell)}$ are sold at price $p^{(\ell)}$, we get the desired revenue bound. We show each of these parts in a separate lemma along with some useful propositions.

For the sake of completeness, our actual algorithm goes as follows: \emph{Run the core algorithm with input allocation $\vec{A}$ and $\gamma=\log(m)$. If $\theta(S') > m$, then return $p^{(0)}$ as the price for the simultaneous mechanism and allocation $\vec{S}^{(0)}$; else return $p^{(\ell)}$ and $\vec{S}^{(\ell)}$.}

\begin{proposition}
\label{prop_surpdiff}
For every $1 \leq t \leq \ell-1$, we have that

$$SW(\vec{B}^{(t)}) - SW(\vec{B}^{(t+1)}) \leq \frac{1}{2 \gamma} SW(\vec{B}^{(t)}).$$
\end{proposition}
\begin{proof}
Fix some value of $t$ in the given range. We know that $p^{(t)} = \frac{SW(\vec{B}^{(t)})}{2\gamma \theta(\vec{B}^{(t)})}$. At this price, each buyer $j$ maximizes her utility from bundle $S^{(t)}_j$, which in turn implies that $v_j(S^{(t)}_j) - p^{(t)}|S^{(t)}_j| \geq v_j(B^{(t)}_j) - p^{(t)}|B^{(t)}_j|$. Rearranging this inequality, and summing up the difference in welfare over all buyers, we get that

$$SW(\vec{B}^{(t)}) - SW(\vec{S}^{(t)}) \leq \sum_{j \in \mathcal{N}}p^{(t)} |B^{(t)}_j| = \frac{SW(\vec{B}^{(t)})}{2\gamma \theta(\vec{B}^{(t)})} \theta(\vec{B}^{(t)}).$$

The proof follows, since for all $t < \ell$, $\vec{B}^{(t+1)} = \vec{S}^{(t)}$.
\end{proof}

\begin{proposition}
\label{prop_priceincrease}
For every $1 \leq t \leq \ell-1$, if $t+1 \neq \gamma$, we have that

$$p^{(t+1)}\geq p^{(t)}.$$
\end{proposition}

\begin{proof}
By definition of $\ell$, we know that for every $t < \ell$, $\theta(\vec{B}^{(t)}) > 2\theta(\vec{S}^{(t)}) = 2\theta(\vec{B}^{(t+1)})$.

Next, let us consider the case when $t < \ell < \gamma$. Then,

$$p^{(t+1)} = \frac{SW(\vec{B}^{(t+1)})}{2\gamma \theta(\vec{B}^{(t+1)})} \geq (1-\frac{1}{2\log m})\frac{SW(\vec{B}^{(t)})}{2\gamma \frac{\theta(\vec{B}^{(t)})}{2}} \geq \frac{SW(\vec{B}^{(t)})}{2\gamma \theta(\vec{B}^{(t)})} = p^{(t)}.$$

The first inequality follows from Proposition~\ref{prop_surpdiff}, and the second one comes from the fact that $(1-\frac{1}{2\log m}) \geq \frac{1}{2}$.
\end{proof}

Note that the reason we do not include $p^{(\gamma)}$ in the above proposition is that the price $p^{(\gamma)}$ is selected according to a special tail selection rule (see Figure \ref{fig:tail}), instead of the usual $Prices$ function.

\begin{lemma}
\label{lem_validitysoln}
For every $1 \leq t \leq \ell$, we have that $\theta(\vec{S}^{(t)}) \leq m$, and thus $(p^{(t)}, \vec{S}^{(t)})$ is a valid outcome of the simultaneous mechanism.

\end{lemma}
\begin{proof}
Since we are in the $\theta(\vec{S'}) \leq m$ regime, we know that at price $p' =  Prices(\vec{A}, \gamma) = p^{(1)}$, the demand for the good is not larger than $m$. Since due to Proposition \ref{prop_priceincrease} the prices are only increasing, then the demand cannot go up since the valuations are monotone, and thus for all $1 \leq t \leq \ell-1$, $\theta(\vec{S}^{(t)}) = \theta(\vec{B}^{(t+1)}) \leq m$. The same also holds for $\theta(\vec{S}^{(\ell)})$ if $\ell<\gamma$, since then Proposition \ref{prop_priceincrease} states that $p^{(\ell)}\geq p^{(\ell-1)}$.


%
%

To conclude, suppose that $\ell=\gamma$. In this case, $p^{(\gamma)}$ is selected according to a special tail selection rule (see Figure \ref{fig:tail}), instead of the usual $Prices$ function: it is the largest single-item valuation of any buyer. Note that since by definition of $\ell$, the number of items sold is at most half of what it was in the previous iteration (i.e., $\theta(\vec{B}^{(t+1)}) <\frac{\theta(\vec{B}^{(t)})}{2}$), and since $\gamma = \log_2 m$, then it must be that when $\ell=\gamma$ then $\theta(\vec{B}^{(\gamma)}) = 1$. In this event, it is not hard to verify that running the simultaneous mechanism with price $p^{(\gamma)}$ will not yield an invalid solution, and the proof follows.
\end{proof}

\begin{lemma}
\label{lemma_lowerboundBell}
$$SW(\vec{B}^{(\ell)}) \geq \frac{SW(\vec{A})}{2} > 0.$$
\end{lemma}
\begin{proof}
The proof is not hard to see. Without loss of generality, suppose that $SW(\vec{A}) \geq SW(\vec{B}^{(t)}$ for all $1 \leq t \leq \ell$ (if this were not the case, we just use the $\vec{B}^{(t)}$ with maximum welfare instead of $\vec{A}$ and argue the same as below). Applying Proposition~\ref{prop_surpdiff} repeatedly from $t=1$ up to $t=\ell-1$ as a telescoping summation, we have that $SW(\vec{B}^{(1)}) - SW(\vec{B}^{(\ell)} \leq \sum_{t=1}^{\ell}\frac{SW(\vec{B}^{(t)}}{2\gamma} \leq \frac{SW(\vec{A})}{2}$. The proof follows since $\vec{B}^{(1)} = \vec{A}$.
\end{proof}

%
%
%


%
%

Now, we are ready to complete the proof. We proceed in two cases based on whether or not $\ell < \gamma$. Suppose that $\ell < \gamma$. We know from Lemma~\ref{lemma_lowerboundBell} that $SW(\vec{B}^{(\ell)}) \geq \frac{SW(\vec{A})}{2}.$ Moreover, by definition of $\ell$, $\theta(\vec{S}^{(\ell)}) \geq \frac{\theta(\vec{B}^{(\ell)})}{2}$. So, we get that,

$$Rev(p^{(\ell)}, \vec{S}^{(\ell)}) = p^{(\ell)}\theta(\vec{S}^{(\ell)}) \geq \frac{SW(\vec{B}^{(\ell)})}{2\gamma\theta(\vec{B}^{(\ell)})}\frac{\theta(\vec{B}^{(\ell)})}{2} \geq \frac{SW(\vec{A})}{8\gamma}.$$

Finally, suppose that $\ell = \gamma$. In this case, it is not hard to deduce that  $\theta(\vec{B}^{(\ell)}) = 1$ since the number of items sold is halved every iteration, and by the special rule setting $p^{(\gamma)}$, we have that $SW(\vec{S}^{(\ell)}) = SW(\vec{B}^{(\ell)})$. The rest of the proof follows from Lemma~\ref{lemma_lowerboundBell}.



\end{document}